\documentclass[10pt,conference]{IEEEtran}

\usepackage{booktabs} 
 
\usepackage{etoolbox}
\usepackage{cite}

\renewcommand\thesection{\arabic{section}}
\renewcommand\thesubsection{\thesection.\arabic{subsection}}
\renewcommand\thesubsubsection{\thesubsection.\arabic{subsubsection}}

\renewcommand\thesectiondis{\arabic{section}}
\renewcommand\thesubsectiondis{\thesectiondis.\arabic{subsection}}
\renewcommand\thesubsubsectiondis{\thesubsectiondis.\arabic{subsubsection}}

\providebool{techreport}
\setbool{techreport}{false}

\usepackage{array}
\usepackage{graphicx}
\usepackage{amsmath}
\usepackage{amssymb}
\usepackage{setspace}
\usepackage{algorithm}
\usepackage[]{algorithmic}
\usepackage{tikz}
\usepackage{xcolor}
\usepackage{multirow}
\usepackage{paralist}
\usepackage{url}
\usepackage{etoolbox}
\usepackage{cancel}
\usepackage{tabu}

\usepackage{mathtools}
\usepackage{marvosym}

\usetikzlibrary{calc,shadows,patterns,shapes,arrows,decorations.pathmorphing,backgrounds,positioning,fit,plotmarks}

\newtheorem{lemma}{Lemma}
\newtheorem{theorem}{Theorem}

 \newtheorem{definition}{Definition}

 \def\myendproof{{\ \vbox{\hrule\hbox{%
   \vrule height1.3ex\hskip0.8ex\vrule}\hrule }}\par}
 \newenvironment{proofAppendix}[1]{\noindent{\bf Proof of #1. }}{\myendproof}

\newcommand{\ceiling}[1]{\left\lceil{#1}\right\rceil}
\newcommand{\setof}[1]{\left\{{#1}\right\}}
\newcommand{\set}[2]{\left\{#1\mid #2\right\}}

\newcommand{\tsocs}[0]{\mbox{\emph{TS-OCS}}}

\newcommand{\litmus}[0]{$\text{LITMUS}^{\text{RT}}$}

\usepackage{array}
\newcolumntype{L}[1]{>{\raggedright\let\newline\\\arraybackslash\hspace{0pt}}m{#1}}
\newcolumntype{C}[1]{>{\centering\let\newline\\\arraybackslash\hspace{0pt}}m{#1}}
\newcolumntype{R}[1]{>{\raggedleft\let\newline\\\arraybackslash\hspace{0pt}}m{#1}}

\tikzset{
    task/.style={shade, shading=radial, rectangle,minimum height=.1cm,
        inner color=#1!20, outer color=#1!60!gray},
    task1/.style={task=yellow, minimum width=13mm},
    task2/.style={task=orange, minimum width=13mm},
    task3/.style={task=red, minimum width=13mm},
    task4/.style={task=green, minimum width=13mm},
    task5/.style={task=blue, minimum width=13mm},
    task6/.style={task=purple, minimum width=13mm},
    task7/.style={task=cyan, minimum width=13mm},
    task8/.style={task=pink, minimum width=13mm},
}

\floatname{algorithm}{Algorithm}
\newcommand{\algorithmicinput}{\textbf{Input:}}
\newcommand{\INPUT}{\item[\algorithmicinput]}

\begin{document}

\title{Dependency Graph Approach for Multiprocessor Real-Time Synchronization}

\author{Jian-Jia Chen, Georg von der Br\"uggen, Junjie Shi and Niklas Ueter\\
TU Dortmund University, Germany}

\maketitle

\begin{abstract}
  Over the years, many multiprocessor locking protocols have been
  designed and analyzed. However, the performance of these protocols
  highly depends on how the tasks are partitioned and prioritized and
  how the resources are shared locally and globally.  This paper
  answers a few fundamental questions when real-time tasks share
  resources in multiprocessor systems. We
  explore the fundamental difficulty of the multiprocessor
  synchronization problem and show that a very simplified version of
  this problem is ${\mathcal NP}$-hard in the strong sense regardless
  of the number of processors and the underlying scheduling paradigm.
  Therefore, the allowance of preemption or migration does not reduce
  the computational complexity.  For the positive side, we develop a
  dependency-graph approach, that is specifically useful for
  frame-based real-time tasks, in which all tasks have the same period
  and release their jobs always at the same time. We present a series
  of algorithms with speedup factors between $2$ and $3$ under
  semi-partitioned scheduling. We further explore methodologies and
  tradeoffs of preemptive against non-preemptive scheduling algorithms
  and partitioned against semi-partitioned scheduling algorithms. The
  approach is extended to periodic tasks under certain conditions.
\end{abstract}

\section{Introduction}
\label{sec:intro}

In a multi-tasking system, mutual exclusion for the accesses to shared
resources, e.g., data structures, files, etc., has to be guaranteed to
ensure the correctness of these operations. Such accesses to shared
resources are typically done within the so-called \emph{critical
  sections}, which can be protected by using \emph{binary semaphores}
or \emph{mutex locks}. Therefore,  at any point in time no two
task instances are in their critical sections that access the same shared
recourse. Moreover, advanced embedded computing systems heavily interact
with the physical world, and \emph{timeliness} of computation is an
essential requirement of correctness. To ensure safe operations of
such embedded systems, the satisfaction of the real-time requirements,
i.e., worst-case timeliness, needs to be verified.

If aborting or restarting a critical section is not allowed, due to
mutual exclusion, a higher-priority job may have to be stopped until a
lower-priority job unlocks the requested shared resource that was
already locked earlier, a so-called priority inversion. The study of
mutual exclusion in uniprocessor real-time systems can be traced back
to the priority inheritance protocol (PIP) and priority ceiling
protocol (PCP) by Sha et al.~\cite{DBLP:journals/tc/ShaRL90} in 1990
and the stack resource policy (SRP) by
Baker~\cite{DBLP:journals/rts/Baker91} in 1991. The Immediate PCP, a
variant of the PCP, has been implemented in Ada (called Ceiling
locking) and POSIX (called Priority Protect Protocol).

To schedule real-time tasks on multiprocessor
platforms, there have been three widely adopted paradigms:
partitioned, global, and semi-partitioned scheduling. The
\emph{partitioned} scheduling approach partitions the tasks statically
among the available processors, i.e., a task is always executed on
the assigned processor.  The \emph{global} scheduling approach allows
a task to migrate from one processor to another at any time. The
\emph{semi-partitioned} scheduling approach decides whether a task is
divided into subtasks statically and how each task/subtask is then
assigned to a processor.  A comprehensive survey of multiprocessor
scheduling in real-time systems can be found in
\cite{DBLP:journals/csur/DavisB11}.

The design of
synchronization protocols for real-time tasks on multiprocessor platforms
 started with the distributed priority ceiling protocol
(DPCP)~\cite{DBLP:conf/rtss/RajkumarSL88}, followed by the multiprocessor priority ceiling protocol
(MPCP)~\cite{Rajkumar_1990}.\footnote{Neither of these
  two protocols had a concrete name in the original papers. In the
  literature, most authors referred to the protocols
  in~\cite{DBLP:conf/rtss/RajkumarSL88} as DPCP and
  \cite{Rajkumar_1990} as MPCP, respectively. } 
The MPCP is based on partitioned fixed-priority scheduling and adopts the PCP for local
resources. When requesting global resources that are shared by several
tasks on different processors, the MPCP executes the
corresponding critical sections with priority boosting. 
By contrast, under the DPCP, the sporadic/periodic
real-time tasks are scheduled based on partitioned fixed-priority
scheduling, except when accessing resources that are bound to a
different processor. That is, the DPCP is semi-partitioned scheduling that
allows migration at the boundary of critical and non-critical sections.

Over the years, many locking protocols have been designed and
analyzed, including the multiprocessor stack resource policy (MSRP)
\cite{DBLP:conf/rtss/GaiLN01}, the flexible multiprocessor locking
protocol (FMLP) \cite{block-2007}, the multiprocessor
PIP \cite{DBLP:conf/rtss/EaswaranA09}, the $O(m)$ locking protocol
(OMLP) \cite{DBLP:conf/rtss/BrandenburgA10}, the Multiprocessor
Bandwidth Inheritance (M-BWI) \cite{DBLP:conf/ecrts/FaggioliLC10},
gEDF-vpr \cite{DBLP:journals/rts/AnderssonE10}, LP-EE-vpr
\cite{DBLP:journals/rts/AnderssonR14}, and the Multiprocessor resource
sharing Protocol (MrsP) \cite{DBLP:conf/ecrts/BurnsW13}.  Also,
several protocols for hybrid scheduling approaches such as
clustered scheduling~\cite{DBLP:conf/ecrts/Brandenburg14},
reservation-based scheduling~\cite{DBLP:conf/ecrts/FaggioliLC10}, and
open real-time systems~\cite{DBLP:conf/ecrts/NematiBN11} have been
proposed in recent years.  To support nested critical sections, Ward
and Anderson~\cite{DBLP:conf/ecrts/WardA12,DBLP:conf/rtns/WardA13}
introduced the Real-time Nested Locking Protocol
(RNLP)~\cite{DBLP:conf/ecrts/WardA12}, which adds supports for
fine-grained nested locking on top of non-nested protocols.

However, the performance of these protocols highly depends on 1) how
the tasks are partitioned and prioritized, 2) how the resources are
shared locally and globally, and 3) whether a job/task being blocked
should spin or suspend itself.

Regarding task partitioning, Lakshmanan et
al.~\cite{DBLP:conf/rtss/LakshmananNR09} presented a
synchronization-aware partitioned heuristic for the MPCP, which
organizes the tasks that share common resources into groups and attempts to
assign each group of tasks to the same processor.  Following the same
principle, Nemati et al.~\cite{DBLP:conf/opodis/NematiNB10} presented
a blocking-aware partitioning method that uses an advanced cost
heuristic algorithm
to split a task group when the entire group fails to be assigned on
one processor.  In subsequent work, Hsiu et
al.~\cite{DBLP:conf/emsoft/HsiuLK11} proposed a dedicated-core
framework to separate the execution of critical sections and normal
sections, and employed a priority-based mechanism for resource
sharing, such that each request can be blocked by at most one
lower-priority request.  Wieder and
Brandenburg~\cite{DBLP:conf/sies/WiederB13} proposed a greedy slacker
partitioning heuristic in the presence of spin locks.  The
\textit{resource-oriented partitioned} (ROP) scheduling was proposed
by Huang et. al~\cite{RTSS2016-resource} in 2016 and later refined by
von der Br\"uggen et al.~\cite{RTNS17-resource} with release
enforcement for a special case. 

For priority assignment, most of the results in the literature use
rate-monotonic (RM) or earliest-deadline-first (EDF) scheduling. To the best of our knowledge, the
priority assignment for systems with shared resources has 
only been seriously explored in a small numbers of papers,
e.g., relative deadline assignment under release enforcement
in \cite{RTNS17-resource}, priority assignment for spinning
\cite{DBLP:conf/rtcsa/AfsharBBN17}, reasonable priority assignments
under global scheduling \cite{DBLP:conf/rtss/EaswaranA09}, and the
optimal priority assignment used in the greedy slack algorithm in
\cite{DBLP:conf/sies/WiederB13}.  However, no theoretical evidence has
been provided to quantify the non-optimality of the above heuristics.

Although many multiprocessor locking protocols have been proposed in
the literature, there are a few unsolved fundamental questions when
real-time tasks share resources (via locking mechanisms) in
multiprocessor systems: 
\begin{itemize}
\item \emph{What is the fundamental difficulty?}
\item \emph{What is the performance gap of partitioned,
    semi-partitioned, and global scheduling?}
\item \emph{Is it always beneficial to prioritize critical sections?}
\end{itemize}

To answer the above questions, we focus on the  simplest and
the most basic 
setting: 
all tasks have the same period and release their jobs
always at the same time, so-called \emph{frame-based real-time task
  systems}, and are scheduled on $M$ identical (homogeneous)
processors. Specifically, we assume that each critical section is
non-nested and is guarded by only one binary semaphore or one mutex
lock.

\noindent\textbf{Contribution:}
Our contributions are as follows:
\begin{compactitem}
\item We show that finding a schedule of the tasks to meet the given
  common deadline is ${\cal NP}$-hard in the strong sense
  \emph{regardless of the number of processors $M$ in the
    system}. Therefore, there is no polynomial-time approximation
  algorithm that can bound the allocated number of processors to meet
  the given deadline. Moreover, the ${\cal NP}$-hardness holds under
  any scheduling paradigm.  Therefore, the allowance of preemption or
  migration does not reduce the computational complexity.
\item We propose a dependency graph approach for multiprocessor
  synchronization, which consists of two steps: 1) the construction of
  a directed acyclic graph (DAG), and 2) the scheduling of this DAG. We
  prove that for minimizing the makespan the lower bound of the approximation
  ratio of such an approach  is at least
  $2-\frac{2}{M}+\frac{1}{M^2}$ under any scheduling paradigm and
  $2-\frac{1}{M}$ under partitioned or semi-partitioned scheduling.
\item We demonstrate how existing results in the literature of
  uniprocessor non-preemptive scheduling  can be adopted to construct the DAG in
  the first step of the dependency graph approach when each
  task has only one critical section. This results in several
  polynomial-time scheduling algorithms with different constant
  approximation bounds for minimizing the makespan. Specifically, the
  best approximation developed is a polynomial-time
  approximation scheme with an approximation ratio of
  $2+\epsilon-\frac{1+\epsilon}{M}$ for any $\epsilon > 0$ under
  semi-partitioned scheduling strategies. We further discuss
  methodologies and tradeoffs of preemptive against non-preemptive
  scheduling algorithms 
  and partitioned against semi-partitioned
  scheduling algorithms.
\item We also implemented the dependency graph approach as a prototype
  in \litmus~\cite{calandrino2006litmus,bbb-2011}. The experimental
  results show that the overhead is almost the same as the
  state-of-the-art multiprocessor locking protocols. Moreover, we also
  provide extensive numerical evaluations, which demonstrate the
  performance of the proposed approach under different scheduling
  constraints. Comparing to the state-of-the-art resource-oriented
  partitioned (ROP) scheduling, our approach shows
  significant improvement.
\end{compactitem}

\section{System Model}

\subsection{Task Model}
In this paper, we will implicitly consider \emph{frame-based real-time
  task systems} to be scheduled on $M$ identical (homogeneous)
processors. The given tasks release their jobs at the same time and
have the same period and relative deadline.  Our studied problem is the task synchronization problem where all
tasks have exactly one (not nested) critical section, denoted as \tsocs.
Specifically, each task $\tau_i$ releases a job (at time $0$ for notational brevity) with the
following properties:
\begin{itemize}
\item $C_{i,1}$ is the execution time of the first non-critical
  section of the job.  
\item $A_{i,1}$ is the execution time of the (first) critical section of
  the job, in which a binary semaphore or a mutex $\sigma(\tau_{i,1})$
  is used to control the access of the critical section.
\item $C_{i,2}$ is the execution time of the second non-critical
  section of the job.  
\end{itemize}
A subjob is a critical section or a non-critical section. Therefore,
there are three subjobs of a job of task $\tau_i$.
We assume the task set $\textbf{T}$ is given and that the deadline
is either implicit, i.e., identical to the period, or constrained, i.e., smaller
than the period. 
The cardinality of a set $\textbf{X}$ is $|\textbf{X}|$.
We
also make the following assumptions:
\begin{itemize}
\item 
For each task $\tau_i$ in $\textbf{T}$, $C_{i,1} \geq 0$, $C_{i,2} \geq 0$, and $A_{i,1} \geq
0$.
\item The execution of the critical sections guarded by one binary semaphore $s$ must be sequentially
executed under a total order.  That is, if two tasks share the same
semaphore, their critical sections must be executed one after another
without any interleaving.
\item The execution of a job cannot be parallelized, i.e., a
  job must be sequentially executed in the order of $C_{i,1}, A_{i,1},
  C_{i,2}$. 
\item There are in total $z$ binary semaphores. 
\end{itemize}
The paper will implicitly focus on the above task model. In
Section~\ref{sec:periodic-tasks}, we will explain how the algorithms
in this paper can be extended to periodic task systems under certain
conditions.

\subsection{Scheduling Strategies}
\label{sec:schedule-strategies}

Here, we define scheduling strategies and the properties of a schedule
for a frame-based real-time task system. Note that the terminologies
used here are 
limited to the scenario where each task in
$\textbf{T}$ releases \emph{only one} job at time $0$. Therefore, we
will use the term jobs and tasks 
interchangeable.   

A schedule is an assignment of the given jobs (tasks) to one of the
$M$ identical processors, such that each job is executed (not
necessarily consecutively) until completion.  A schedule for
$\textbf{T}$ can be defined as a function $\rho: {\mathbb R}\times M
\rightarrow \textbf{T}\cup\setof{\bot}$, where $\rho(t,m) = \tau_j$
denotes that the job of task $\tau_j$ is executed at time $t$ on
processor $m$, and $\rho(t, m) = \bot$ denotes that processor $m$ is
idle at time $t$.  We assume that a job has to be sequentially
executed, i.e., intra-task parallelism is not possible. Therefore, it is not
feasible to run a job in parallel on two processors, i.e., $\rho(t, m) \neq \rho(t, m')$ for any $m \neq
m'$ if $\rho(t, m) \neq \bot$.

Some other constraints may also be introduced. A schedule is
\emph{non-preemptive} if a job cannot be preempted by any other
job, i.e., there is only one interval with $\rho(t,m)~=~\tau_j$ on processor
$m$ for each task $\tau_j$ in $\textbf{T}$.  A schedule is
\emph{preemptive} \index{preemptive} if a job can be preempted,
i.e., more than one interval with $\rho(t,m)~=~\tau_j$ for any task
$\tau_j$ in $\textbf{T}$ on processor $m$ is allowed.

For a \emph{partitioned} schedule, a job has to be executed on one
processor. That is, there is exactly one processor $m$ with
$\rho(t,m)~=~\tau_j$ for every task $\tau_j$ in $\textbf{T}$.  Such a
schedule can be preemptive or non-preemptive.  For a \emph{global
  schedule}, a job can be arbitrarily executed on any of the $M$
processors at any time point. That is, it is possible that
$\rho(t,m)~=~\tau_j$ and $\rho(t', m')~=~\tau_j$ for $m\neq m'$ and $t
\neq t'$. By definition, a global schedule is preemptive (for
frame-based real-time task systems) in our model. 
For a
\emph{semi-partitioned} schedule, a subjob (either a critical section
or a non-critical section) has to be executed on one processor. Such a
semi-partitioned schedule can be preemptive or non-preemptive.

Based on the above definitions, a partitioned schedule is also a
semi-partitioned schedule, and a semi-partitioned schedule is also a
global schedule.

\subsection{Scheduling Theory}
\label{sec:scheduling-theory}
  In the rich literature of scheduling theory, one
specific objective is to minimize the completion time of the jobs,
called \textbf{makespan}. For frame-based real-time task systems,
if the makespan of the jobs released at time $0$ is no more than the
relative deadline,
then the task set can be feasibly scheduled to meet the
deadline.\footnote{Note that the deadline is never larger than the period in
our setting.} 
 We state the makespan problem for \tsocs\ that is studied here as follows:

\begin{definition}
  \label{def:problem-def}
  \textbf{The  \tsocs\ Makespan Problem:} 
  We are given $M$ identical (homogeneous) processors. There are $N$
  tasks arriving at time $0$. Each task is given by $\setof{C_{i,1},
    A_{i,1}, C_{i,2}}$ and has at most one critical section, guarded
  by one binary semaphore. The objective is to find a schedule that
  minimizes the makespan.
\end{definition}
Alternatively, we can also investigate the \textbf{bin packing}
version of the problem, i.e., minimizing the number of allocated
processors to meet a given common deadline $D$.

\begin{definition}
  \label{def:problem-bin-packing-def}
  \textbf{The \tsocs\  Bin Packing Problem:} 
  We are given identical (homogeneous)
  processors. There are $N$ tasks arriving at time $0$ with a common
  deadline $D$. Each task is given by $\setof{C_{i,1}, A_{i,1}, C_{i,2}}$
  and has at most one critical section, guarded by one binary
  semaphore. The objective is to find a schedule to meet the deadline
  with the minimum number of allocated processors.
\end{definition}

Essentially, the decision versions of the makespan and the bin packing
problems are identical: 
\begin{definition}
  \label{def:problem-decision-version-def}
  \textbf{The \tsocs\  Schedulability Problem:} 
  We are given $M$ identical (homogeneous)
  processors. There are $N$ tasks arriving at time $0$ with a common
  deadline $D$. Each task is given by $\setof{C_{i,1}, A_{i,1}, C_{i,2}}$
  and has at most one critical section, guarded by one binary
  semaphore. The objective is to find a schedule to meet the deadline
  by using the $M$ 
  processors.
\end{definition}

In the domain of scheduling theory, a scheduling problem is described
by a triplet $\mbox{Field}_1|\mbox{Field}_2|\mbox{Field}_3$.
\begin{itemize}
\item $\mbox{Field}_1$: describes the machine environment. 
\item $\mbox{Field}_2$: specifies the processing characteristics and
  constraints. 
\item $\mbox{Field}_3$: presents the objective to be optimized.
\end{itemize}
For example, the scheduling problem $1|r_j|L_{\max}$ deals with a
uniprocessor system, in which the input is a set of jobs with
different release times and different absolute deadlines, and the
objective is derive a non-preemptive schedule which minimizes the
maximum lateness.  The scheduling problem $P||C_{\max}$ deals with a
homogeneous multiprocessor system, in which the input is a set of jobs
with the same release times, and the objective is derive a
\emph{partitioned} schedule which minimizes the makespan.  The
scheduling problem $P|prec|C_{\max}$ is an extension of $P||C_{\max}$
by further considering the precedence constraints of the jobs.  The
scheduling problem $P|prec,prmp|C_{\max}$ further allows preemption. 
Note that in classical scheduling theory, preemption in parallel
machines implies the possibility of job migration from one machine to
another.\footnote{In real-time systems, this is
not necessarily the case. For instance, under preemptive partitioned scheduling
a job can be preempted and resumed later on the same processor without
migration.
} Therefore, the scheduling problem $P|prec,prmp|C_{\max}$ allows job preemption and migration, i.e., preemptive global scheduling.

\subsection{Approximation Metrics}
\label{sec:approximation-metrics}

Since many scheduling problems are ${\cal NP}$-hard in the strong
sense, polynomial-time approximation algorithms are often used. In the
realm of real-time systems, there are two widely adopted metrics: 

The \emph{Approximation Ratio} compares the resulting objectives of
(i) scheduling algorithm ${\cal A}$ and (ii) an optimal algorithm
when scheduling
any given task set.  Formally, an algorithm ${\cal A}$ for
the makespan problem (i.e., Definition~\ref{def:problem-def}) has an
approximation ratio $\alpha\geq 1$, if given any task set $\textbf{T}$, the
resulting makespan is at most $\alpha C_{\max}^*$ on $M$ processors,
where $C_{\max}^*$ is the minimum (optimal) makespan to schedule
$\textbf{T}$ on $M$ processors.  An algorithm ${\cal A}$ for the bin
packing problem (i.e., Definition~\ref{def:problem-bin-packing-def})
has an approximation ratio $\alpha\geq 1$, if given any task set
$\textbf{T}$, it can find a schedule of $\textbf{T}$ on $\alpha M^*$
processors to meet the common deadline, where $M^*$ is the minimum
(optimal) number of processors required to feasibly schedule $\textbf{T}$.

The \emph{Speedup Factor}~\cite{Kalyanasundaram:2000,Phillips:stoc97}
of a scheduling algorithm ${\cal A}$ indicates the factor $\alpha\geq1$ by
which the overall speed of a system would need to be increased so that
the scheduling algorithm ${\cal A}$ always derives a feasible schedule
to meet the deadline, provided that there exists one at the original
speed. This is used for the problem in 
Definition~\ref{def:problem-decision-version-def}.

We note that an algorithm that has an approximation ratio $\alpha$ for
the makespan problem in Definition~\ref{def:problem-def} also has a
speedup factor $\alpha$ for the schedulability problem in
Definition~\ref{def:problem-decision-version-def}. 

\section{Dependency Graph Approach for Multiprocessor Synchronization}

To handle the studied makespan problem in
Definition~\ref{def:problem-def}, we propose a \textbf{Dependency
  Graph Approach}, which involves two steps:
  \begin{itemize}
  \item In the first step, a directed graph $G = (V, E)$ is
    constructed. A subjob (i.e., a critical section or a non-critical
    section) is a vertex in $V$. The subjob $C_{i,1}$ is a
    predecessor of the subjob $A_{i,1}$. The subjob $A_{i,1}$ is
    a predecessor of the subjob $C_{i,2}$.  If two jobs of $\tau_i$
    and $\tau_j$ share the same binary semaphore, i.e.,
    $\sigma(\tau_{i,1}) = \sigma(\tau_{j,1})$, then either the subjob
    $A_{i,1}$ is the predecessor of that of $A_{j,1}$ or the subjob
    $A_{j,1}$ is the predecessor of that of $A_{i,1}$. All the
    critical sections guarded by a binary semaphore form a chain in
    $G$, i.e., the critical sections of the binary semaphore follow a
    total order. Therefore, we have the following properties in set $E$:
    \begin{itemize}
    \item The two directed edges $(C_{i,1}, A_{i,1})$ and $(A_{i,1},
      C_{i,2})$ are in $E$. 
    \item Suppose that $\textbf{T}_k$ is the set of the tasks which require
      the same binary semaphore $s_k$. Then, the $|\textbf{T}_k|$ tasks in
      $\textbf{T}_k$ follow a certain total order $\pi$ such that
      $(A_{i,1}, A_{j,1})$ is a directed edge in $E$ when $\pi(\tau_i) =
      \pi(\tau_j)-1$.
    \end{itemize} 
    Fig.~\ref{fig:example-dependency-graph} provides an example for
    a task dependency graph with one binary semaphore. Since there are
    $z$ binary semaphores in the task set, the task dependency graph
    $G$ has in total $z$ connected subgraphs, denoted as $G_1, G_2,
    \ldots, G_z$.
     In each connected subgraph $G_\ell$, the
    corresponding critical sections of the tasks that request critical
    sections guarded by the same semaphore form a chain and have to be
    executed sequentially. For example, in
    Fig.~\ref{fig:example-dependency-graph}, the dependency graph
    forces the scheduler to execute the critical section $A_{1,1}$
    prior to any of the other three critical sections.
  \item In the second step, a corresponding schedule of $G$ on $M$
    processors is generated. The schedule can be based on system's restrictions or user's preferences, i.e., either preemptive or
    non-preemptive schedules, either global, semi-partitioned, or partitioned
    schedules.
  \end{itemize}

  \begin{figure}[t]
    \centering
  \scalebox{0.8}{
      \begin{tikzpicture}[x=0.75cm,auto, thick]
    \node [draw,circle](C11)at(0,4){$C_{1,1}$};
    \node [draw,rectangle](A11)at(0,2){$A_{1,1}$};
   \node [draw,circle](C12)at(0,0){$C_{1,2}$};
    \node [draw,circle](C21)at(3,4){$C_{2,1}$};
    \node [draw,rectangle](A21)at(3,2){$A_{2,1}$};
   \node [draw,circle](C22)at(3,0){$C_{2,2}$};
    \node [draw,circle](C31)at(6,4){$C_{3,1}$};
    \node [draw,rectangle](A31)at(6,2){$A_{3,1}$};
   \node [draw,circle](C32)at(6,0){$C_{3,2}$};
  \node [draw,circle](C41)at(9,4){$C_{4,1}$};
    \node [draw,rectangle](A41)at(9,2){$A_{4,1}$};
   \node [draw,circle](C42)at(9,0){$C_{4,2}$};

    \draw[->] (C11) -- (A11);
    \draw[->] (A11) -- (C12);

    \draw[->] (C21) -- (A21);
    \draw[->] (A21) -- (C22);

    \draw[->] (C31) -- (A31);
    \draw[->] (A31) -- (C32);

    \draw[->] (C41) -- (A41);
    \draw[->] (A41) -- (C42);

    \draw[->,very thick] (A11) -- (A21);
    \draw[->,very thick] (A21) -- (A31);
    \draw[->,very thick] (A31) -- (A41);

  \end{tikzpicture}}     
    \caption{An example of a task dependency graph for a task set with
      one binary semaphore.}
    \label{fig:example-dependency-graph}
  \end{figure}
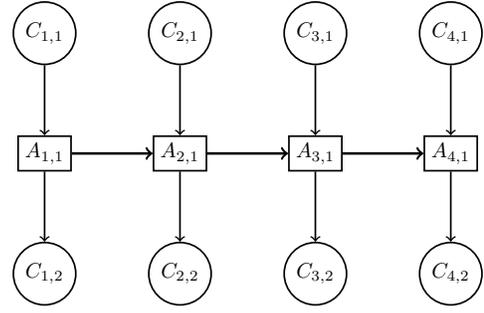

  In the dependency graph approach, the second step has been widely
  studied in scheduling theory. That is, a solution of the problem
  $P|prec|C_{\max}$ results in a semi-partitioned schedule, since the
  dependency graph is constructed by considering a critical section or
  a non-critical section as a subjob. Moreover, a solution of the
  problem $P|prec,prmp|C_{\max}$ results in a global schedule. For
  deriving a partitioned schedule, we can force the subjobs generated
  by a job to be \emph{tied} to one processor. That is, $P|prec,
  tied|C_{\max}$ targets a partitioned non-preemptive schedule and
  $P|prec, prmp, tied|C_{\max}$ targets a partitioned preemptive
  schedule.  

  Therefore, the key issue is the construction of the dependency
  graph. An alternative view of the dependency graph approach is to
   build the dependency graph assuming  
   a sufficient number of 
  processors (i.e., using as many processors as possible) in the first
  step, and then the second step considers the constraint of the
  number of processors. Towards the first step, we need the following
  definition:
  \begin{definition}
    A \textbf{critical path} of a task dependency graph $G$ is one of the
    longest paths of $G$. The critical path length of $G$ is denoted by $len(G)$
  \end{definition}

  For the rest of this paper, we denote a dependency task graph of the
  input task set $\textbf{T}$ that has the minimum critical path length as
  $G^*$. Note that $G^*$ is independent of $M$.
  \begin{lemma}
    \label{lemma:critical-path-lb}
    $len(G^*)$ is the lower bound of the \tsocs{} makespan problem for task set
    $\textbf{T}$ on $M$ processors.
  \end{lemma}
  \begin{proof}
    This comes from the setting of the problem, i.e., each task
    $\tau_i$ has only one critical section guarded by one binary
    semaphore, and the definition of the graph $G^*$, i.e., using as
    many processors as possible.
  \end{proof}

  \begin{definition}
    A \textbf{feasible schedule $S(G)$} of a task dependency graph $G$ respect
    to the precedence constraints defined in $G$ and the specified
    scheduling requirement, e.g., being
    global/semi-partitioned/partitioned and preemptive/non-preemptive.
    $L(S(G))$ is the makespan of $S(G)$.
  \end{definition}

  With the above definitions, we can recap the objectives of the two
  steps in the dependency graph approach. In the first step, we would
  like to construct a dependency graph $G$ to minimize $len(G)$, and
  in the second step, we would like to construct a schedule $S(G)$ to
  minimize $L(S(G))$.

  We conclude this section by stating the following theorem:
\begin{theorem}
    \label{theorem:lower-bound}
    The optimal makespan of the \tsocs\ makespan problem 
	for $\textbf{T}$ on $M$
    processors is at least
    \begin{equation}\label{eq:lb}
    \max\left\{\sum_{\tau_i \in \textbf{T}} \frac{C_{i,1}+A_{i,1}+C_{i,2}}{M},
    len(G^*)\right\}
    \end{equation}
    where $G^*$ is a dependency task graph of $\textbf{T}$ that has
    the minimum critical path length.
  \end{theorem}
  \begin{proof}
    The lower bound $len(G^*)$ comes from
    Lemma~\ref{lemma:critical-path-lb} and the lower bound
    $\sum_{\tau_i \in \textbf{T}} \frac{C_{i,1}+A_{i,1}+C_{i,2}}{M}$ is due to the
    pigeon hole principle.
  \end{proof}

\section{Computational Complexity and Lower Bounds}

This section presents the computational complexity and lower bounds of
approximation ratios of the dependency graph approach.

\subsection{Computational Complexity}

The following theorem shows that constructing $G^*$ is unfortunately
${\cal NP}$-hard in the strong sense.
  \begin{theorem}
    \label{theorem:critical-path-np-hard}
    Constructing a dependency task graph $G^*$ that has the minimum
    critical path length is ${\cal NP}$-hard in the strong sense.
  \end{theorem}
  \begin{proof}
    This theorem is proved by a reduction from the decision version of
    the scheduling problem $1|r_j| L_{\max}$, i.e., uniprocessor
    non-preemptive scheduling, in which the objective is to minimize
    the maximum lateness 
    assuming that each job $J_j$ in the given job set
    $\textbf{J}$ has 
    its known processing time $p_j \geq 0$, arrival time
    $r_j \geq 0$, and absolute deadline $d_j$. This problem is ${\cal
      NP}$-hard in the strong sense by a reduction from the
    3-Partition problem~\cite{LenstraRKBr77}.  Suppose that the
    decision version of the scheduling problem $1|r_j|
    L_{\max}$ is to validate whether there exists a schedule in which
    the finishing time of each job $J_j$ is no less than $d_j$.

    Let $H$ be any positive integer greater than $\max_{j \in \textbf{J}}
    d_j$.  For each job $J_j$ in $\textbf{J}$, we construct a task
    $\tau_j$ with one critical section, where $C_{j,1}$ is set to
    $r_j$, $C_{j,2}$ is set to $H - d_j$, and $A_{j,1} $ is set to
    $p_j$. By the setting, $C_{j,1} \geq 0, C_{j,2} \geq 0$, and
    $A_{j,1} \geq 0$ for every constructed task $\tau_j$.  The
    critical sections of all the constructed tasks are guarded by
    \emph{only one} binary semaphore. Let the task set constructed
    above be $\textbf{T}$. The above input task set $\textbf{T}$ by
    definition is a feasible input task set for the
    one-critical-section task synchronization problem.

    We now prove that there is a non-preemptive uniprocessor schedule
    for $\textbf{J}$ in which all the jobs can meet their deadlines if
    and only if there is a dependency task graph $G^*$ with a critical
    path length less than or equal to $H$ for the constructed task set
    $\textbf{T}$.

    \textbf{If part}, i.e., $len(G^*) \leq H$ holds: Without loss of
    generality, we index the tasks in $\textbf{T}$ so that the
    critical section of $A_{i,1}$ is the immediate predecessor of the
    critical section $A_{i+1,1}$ in $G^*$, e.g., as in
    Fig.~\ref{fig:example-dependency-graph}.  Suppose that
    $G^*(\tau_i)$ is the subgraph of $G^*$ that consists of only the
    vertices representing  $\set{C_{k,1}, A_{k,1},
      C_{k,2}}{k=1,2,\ldots,i-1}\cup \setof{C_{i,1}, A_{i,1}}$ and the
    corresponding edges. Let $f_i$ be the longest path in
    $G^*(\tau_i)$ that \emph{ends at the vertex representing 
      $A_{i,1}$}.

    By definition, $f_1$ is $C_{1,1} + A_{1,1}$. Moreover, $f_i$ is
    $\max\{f_{i-1}, C_{i,1}\} + A_{i,1}$ for every task $\tau_i$ in
    $\textbf{T}$.  Since $len(G^*) \leq H$ and $C_{i,2} = H - d_i$, we
    know that $f_i + C_{i,2} \leq H \Rightarrow f_i \leq d_i$ for
    every task $\tau_i$ in $\textbf{T}$. 

    We can now construct the uniprocessor non-preemptive schedule for
    $\textbf{J}$ by following the same execution order. Here, we index
    the jobs in $\textbf{J}$ corresponding to $\textbf{T}$. The
    finishing time of job $J_1$ is $r_1 + p_1 = C_{1,1} + A_{1,1} =
    f_1$. The finishing time of job $J_i$ is $\max\{f_{i-1}, r_i\} +
    p_i = \max\{f_{i-1}, C_{i,1}\} + A_{i,1} = f_i$.  

    
    This proves the if part.

    \textbf{Only-If part}, i.e., there is a uniprocessor non-preemptive
    schedule in which all the deadlines of the jobs in $\textbf{J}$
    are met: The proof for the if part can be reverted and the same
    arguments can be applied. Due to space limitation, we omit the details.
  \end{proof}

  \begin{theorem}
    \label{theorem:makespan-np-hard}
    The makespan problem with task synchronization for $\textbf{T}$ on
    $M$ processors is ${\cal NP}$-hard in the strong sense even if $M$
    is sufficiently large under any scheduling paradigm. 
  \end{theorem}
  \begin{proof}
    This comes directly from
    Theorem~\ref{theorem:critical-path-np-hard}. Consider that there
    are $M \geq |\textbf{T}|+1$ processors. The if-and-only-if proof
    in Theorem~\ref{theorem:critical-path-np-hard} can be extended by
    introducing a concrete schedule that executes the two non-critical
    sections of task $\tau_i$ one processor $i$ and the critical
    section of task $\tau_i$ on processor
    $|\textbf{T}|+1$.\footnote{The same statement also holds for using
      $M =|\textbf{T}|$ processors, but the proof is more involved. }
  \end{proof}

  Theorem~\ref{theorem:makespan-np-hard} expresses the fundamental
  difficulty of the multiprocessor synchronization problem and shows
  that a very simplified version of this problem is ${\mathcal
    NP}$-hard in the strong sense regardless of the number of
  processors and the underlying scheduling paradigm.  Therefore, the
  allowance of preemption or migration does not reduce the
  computational complexity. The fundamental problem is the sequencing
  of the critical sections, which is independent from the underlying
  scheduling paradigm. Therefore, no matter what flexibility the
  scheduling algorithm has (unless aborting and restarting a critical
  section is allowed), the computational complexity remains 
  ${\mathcal NP}$-hard in the strong sense.

\subsection{Remarks: Bin Packing}

Although the focus of this paper is the makespan problem in
Definition~\ref{def:problem-def} and the schedulability problem in
Definition~\ref{def:problem-decision-version-def}, we also state the
following theorems to explain the difficulty of the bin packing
problem in Definition~\ref{def:problem-bin-packing-def}.

\begin{theorem}
    \label{theorem:bin-np-hard}
    Minimizing the number of processors for a given common deadline of
    $\textbf{T}$ with task synchronization for $\textbf{T}$ (i.e.,
    Definition~\ref{def:problem-bin-packing-def}) is ${\cal NP}$-hard
    in the strong sense under any scheduling paradigm.
  \end{theorem}
  \begin{proof}
    As the decision problem is
    Definition~\ref{def:problem-decision-version-def}, we reach the
    conclusion based on Theorem~\ref{theorem:makespan-np-hard}.
  \end{proof}

\begin{theorem}
    \label{theorem:bin-np-no-approximation}
    There is no polynomial-time (approximation) algorithm to minimize
    the number of processors for a given common deadline of
    $\textbf{T}$ with task synchronization for $\textbf{T}$ under any
    scheduling paradigm unless ${\cal P}={\cal NP}$.
  \end{theorem}
  \begin{proof}
    This is based on
    Theorems~\ref{theorem:critical-path-np-hard}~and~\ref{theorem:makespan-np-hard}. If
    such a polynomial-time algorithm exists, then the problem
    $1|r_j|L_{\max}$ can be solved in polynomial time, which implies
    ${\cal P}={\cal NP}$.
  \end{proof}

\subsection{Lower Bounds}

The dependency graph approach requires two steps. The following
theorem shows that even if both steps are optimized, the resulting
schedule for the makespan problem with task synchronization is not
optimal and has an asymptotic lower bound $2$ of the approximation
ratio.
  \begin{theorem}
    \label{theorem:non-optimal-approach}
    The optimal schedule on $M$ identical processors for the dependency
    graph $G^*$ that has the minimum critical path length is not
    optimal for the \tsocs\ makespan problem 
     and can have an approximation bound of at least
    \begin{itemize} 
    \item $2-\frac{2}{M}+\frac{1}{M^2}$ under any scheduling paradigm,
      and
    \item $2-\frac{1}{M}$ under partitioned or semi-partitioned scheduling.
    \end{itemize}
  \end{theorem}
  \begin{proof}
    We prove this theorem by providing a concrete input instance as follows:
    \begin{itemize}
    \item Suppose that $M$ is a given integer with $M \geq 2$ and we
      have $N = M^2-M+1$ tasks.
    \item We assume a small positive number $\delta$ which is close to
      $0$ and a number $Q$ which is much greater than $\delta$, i.e.,
      $\frac{Q}{MN}\gg \delta > 0$.
    \item All $N$ tasks have a critical section guarded by the same
      binary semaphore.
    \item Task $\tau_1$ has $C_{1,1} = \delta, A_{1,1}=Q-\frac{Q}{M}$, and
      $C_{1,2} = \frac{Q}{M} +N\delta$
    \item Task $\tau_i$ has $C_{i,1} = \delta, A_{i,1}=\delta$, and
      $C_{i,2} = \frac{Q}{M}$ for \mbox{$i=2,3,\ldots,N$.}
    \end{itemize}
    We need to show that the optimal dependency graph of this input
    instance in fact leads to the specified bound. The proof
    is in Appendix.
  \end{proof}

\section{Algorithms to Construct $G$}
\label{sec:heuristic-derive-G}

The key to success is to find $G^*$.  Unfortunately, as shown in
Theorem~\ref{theorem:critical-path-np-hard}, finding $G^*$ is ${\cal
  NP}$-hard in the strong sense. However, finding good approximations
is possible. The problem to construct $G$ is called the
\emph{dependency-graph construction problem}.  Here, instead of
presenting new algorithms to find good approximations of $G^*$, we
explain how to use the existing algorithms of the scheduling problem
$1|r_j|L_{\max}$ to derive good approximations of $G^*$.

It should be first noted that the problem $1|r_j|L_{\max}$ cannot be
approximated with a bounded approximation ratio because the optimal
schedule may have no lateness at all and any approximation leads to an
unbounded approximation ratio. However, a variant of this problem can
be easily approximated. This is known as the \emph{delivery-time}
model of the problem $1|r_j|L_{\max}$. In this model, each job $J_j$
has its release time $r_j$, processing time $p_j$, and delivery time
$q_j \geq 0$. After a job finishes its execution on a machine, its
result (final product) needs $q_j$ amount of time to be delivered to
the customer. The objective is to minimize the makespan
$K$. Therefore, the \emph{effective} deadline $d_j$ of job $J_j$ on
the given single machine is $d_j = K-q_j$. Since $K$ is a constant,
this is effectively equivalent to the case when $d_j$ is set to
$-q_j$. 

The delivery-time model of the problem $1|r_j|L_{\max}$ can then be
effectively approximated. Moreover, our problem to construct a good
dependency graph for $\textbf{T}$ is indeed equivalent to the
delivery-time model of the problem $1|r_j|L_{\max}$. To show such
equivalence, Algorithm~\ref{alg:graph-construction} presents the
detailed transformation. For each semaphore $s_k$, suppose that
$\textbf{T}_k$ is the set of tasks that use $s_k$ (Line 1 in
Algorithm~\ref{alg:graph-construction}). For each task set
$\textbf{T}_k$, we transform the problem
to construct $G_k$ to an equivalent delivery-time model of the problem
$1|r_j|L_{\max}$ (Line 3 to Line 8). Then, we construct the graph
$G_k$ based on the derived schedule of an approximation algorithm for
the delivery-time model of the problem $1|r_j|L_{\max}$. 

\begin{theorem}
  \label{thm:equivalent-delivery-time+graph}
  An $\alpha$-approximation algorithm for the delivery-time model of
  the problem $1|r_j|L_{\max}$ applied in
  Algorithm~\ref{alg:graph-construction} guarantees to derive a
  dependency graph $G$ with $len(G) \leq \alpha \times len(G^*)$.
\end{theorem}
\begin{proof}
  This theorem can be proved by a counterpart of the proof of
  Theorem~\ref{theorem:critical-path-np-hard}. We will show that
  Algorithm~\ref{alg:graph-construction} is in fact an L-reduction
  (i.e., a reduction that preserves the approximation ratio) from the
  input task set to the delivery-time model of the problem
  $1|r_j|L_{\max}$. In this L-reduction, there is no loss of the
  approximation ratio.

  First, by definition, two tasks are independent if they do not share
  any semaphore. Moreover, since the \tsocs{} problem assumes that a
  task accesses at most one binary semaphore, a task $\tau_i$ can only
  appear at most in one $\textbf{T}_k$ for a certain $k$. Therefore,
  $len(G^*) = \max_{k=1,2,\ldots,z} len(G^*_k)$.

  To show that the reduction preserves the approximation ratio, we
  only need to prove the one-to-one mapping. One possibility is to
  prove that a schedule for the input instance of the problem
  $1|r_j|L_{\max}$ delivers the last result at time $X$ if and only if
  the corresponding graph $G_k$ constructed by using Lines 9 and 10 in
  Algorithm~\ref{alg:graph-construction} has a critical path length
  $X$. This is unfortunately not possible because a (\emph{technically
    bad but possible}) schedule for the input instance of the problem
  $1|r_j|L_{\max}$ can be arbitrarily alerted by inserting useless
  delays.

  Fortunately, for a given permutation to order the $|\textbf{T}_k|$
  tasks in $\textbf{T}_k$, we can always construct a schedule for the
  input instance of the problem $1|r_j|L_{\max}$ by respecting the
  given order and their release times. Such a schedule for the input
  instance of the problem $1|r_j|L_{\max}$ delivers the last result at
  time $X$ if and only if the corresponding graph $G_k$ constructed by
  using Lines 9 and 10 in Algorithm~\ref{alg:graph-construction} has a
  critical path length $X$. Moreover, the schedule for one such
  permutation is optimal for the input instance of the problem
  $1|r_j|L_{\max}$.

  Therefore, the approximation ratio is perserved while constructing
  $G_k$. According to the above discussions, $len(G_k) \leq \alpha\times
  len(G_k^*)$. Moreover, 
  \begin{align*}
len(G) &\leq \;\;\;\;\;\;\max_{k=1,2,\ldots,z} len(G_k)\\
  &\leq \alpha \times \max_{k=1,2,\ldots,z} len(G_k^*) = \alpha\times len(G^*)    
  \end{align*}
\end{proof}

\begin{algorithm}[t]
  \caption{Graph Construction Algorithm}
  \label{alg:graph-construction}
  \begin{algorithmic}[1]
\footnotesize
    \INPUT set ${\bf T}$ of $N$ tasks with $z$ shared binary semaphores;

    \STATE $\textbf{T}_k \leftarrow
    \set{\tau_i}{\sigma(\tau_{i,1})=s_k}$ for $k=1,2,\ldots,z$;

    \FOR {$k \leftarrow 1$ to $z$}

    \STATE $\textbf{J} \leftarrow \emptyset$;

    \FOR {each $\tau_i \in \textbf{T}_k$}

    \STATE create a job $J_i$ with $r_i\leftarrow C_{i,1}$,
    $p_i\leftarrow A_{i,1}$, and $q_i \leftarrow
    C_{i,2}$, where $q_i$ is the delivery time;

    \STATE $\textbf{J} \leftarrow \textbf{J} \cup \setof{J_i}$;

    \ENDFOR

    \STATE apply an approximation algorithm to derive a non-preemptive
    schedule $\rho_k$ for the delivery-time model of the problem
    $1|r_j|L_{\max}$ on one machine;

    \STATE construct the initial dependency graph $G_k$ for
    $\textbf{T}_k$, in which the following directed edges $(C_{i,1},
    A_{i,1})$ and $(A_{i,1}, C_{i,2})$ for every task $\tau_i \in
    \textbf{T}_k$ are created;

    \STATE create a directed edge from $A_{i,1}$ to $A_{j,1}$ in $G_k$
    if job $J_j$ is executed right after (but not necessarily
    consecutively) job $J_i$ in the schedule $\rho_k$;

    \ENDFOR

    \STATE return $G=G_1 \cup G_2 \cup \ldots \cup G_z$;
  \end{algorithmic}
 \end{algorithm}

 According to Theorem~\ref{thm:equivalent-delivery-time+graph} and
 Algorithm~\ref{alg:graph-construction}, we can simply apply the
 existing algorithms of the scheduling problem $1|r_j|L_{\max}$ in the
 delivery-time model to derive $G^*$ by using well-studied
 branch-and-bound methods, see for
 example~\cite{CARLIER198242,doi:10.1287/opre.23.3.475,NOWICKI1986266},
 or good approximations of $G^*$, see for
 example~\cite{DBLP:journals/mor/HallS92,doi:10.1287/opre.28.6.1436}.
 Here, we will summarize several polynomial-time approximation algorithms. The details
 can be found in~\cite{DBLP:journals/mor/HallS92}.
 
 For the delivery-time model of the scheduling problem
 $1|r_j|L_{\max}$, the \textbf{extended Jackson's rule}  (\textbf{JKS}) is as follows:
 ``Whenever the machine is free and one or more jobs is available for
 processing, schedule an available job with largest delivery time,'' as
 explained in~\cite{DBLP:journals/mor/HallS92}.
 
\begin{lemma}
  \label{lemma:uniprocessor-EDD}
 The extended Jackson's rule (\textbf{JKS}) is a polynomial-time
  $2$-approximation algorithm for the dependency-graph construction
  problem.
\end{lemma}
\begin{proof}
  This is based on Theorem~\ref{thm:equivalent-delivery-time+graph}
  and the approximation ratio of \textbf{JKS} for the problem
  $1|r_j|L_{\max}$, where the proof can be found
  in~\cite{Kise1979KJ00001202213}.
\end{proof}

Potts~\cite{doi:10.1287/opre.28.6.1436} observed some nice properties
when the extended Jackson's rule is applied. Suppose that the last
delivery is due to a job $J_c$. Let $J_a$ be the earliest scheduled
job so that the machine in the problem $1|r_j|L_{\max}$ is not idle
between the processing of $J_a$ and $J_c$. The sequence of the jobs
that are executed sequentially from $J_a, \ldots,$ to $J_c$ is called
a \emph{critical sequence}. By the definition of $J_a$, all jobs in
the critical sequence must be released no earlier than the release
time $r_a$ of job $J_a$. If the delivery time of any job in the
critical sequence is not shorter than the delivery time $q_c$ of
$J_c$, then it can be proved that the extended Jackson's rule is
optimal for the problem $1|r_j|L_{\max}$. However, if the delivery
time $q_b$ of a job $J_b$ in the critical sequence is shorter than the
delivery time $q_c$ of $J_c$, the extended Jackson's rule may start a
non-preemptive job $J_b$ too early. Such a job $J_b$ that appears last
in the critical sequence is called the \emph{interference job} of the
critical sequence.

Potts~\cite{doi:10.1287/opre.28.6.1436} suggested to \emph{attempt at
  improving the schedule by forcing some interference job to be
  executed after the critical job $J_c$}, i.e., by delaying the
release time of $J_b$ from $r_b$ to $r_b'=r_c$. This procedure is
repeated for at most $n$ iterations and the best schedule among the
iterations is returned as the solution.
\begin{lemma}
  \label{lemma:uniprocessor-Potts}
  Potts' iterative process (\textbf{Potts})  is a polynomial-time
  $1.5$-approximation algorithm for the dependency-graph construction
  problem.
\end{lemma}
\begin{proof}
 This is based on Theorem~\ref{thm:equivalent-delivery-time+graph}
  and the approximation ratio of \textbf{Potts} for the problem
  $1|r_j|L_{\max}$, where the proof can be found
  in~\cite{DBLP:journals/mor/HallS92}.
\end{proof}

Hall and Shmoys~\cite{DBLP:journals/mor/HallS92} further improved the
approximation ratio to $4/3$ by handling a special case when there are
two jobs $J_i$ and $J_h$ with $p_i > P/3$ and $p_h > P/3$ where $P$
is $\sum_{J_j} p_j$ and running Potts' algorithm for $2n$
iterations.\footnote{Hall and Shmoys~\cite{DBLP:journals/mor/HallS92}
  further use the concept of forward and inverse problems of the input
instance of $1|r_j|L_{\max}$. As they are not highly related, we omit
those details. }
\begin{lemma}
  \label{lemma:uniprocessor-HS}
  Algorithm \textbf{HS} is a polynomial-time $4/3$-approximation
  algorithm for the dependency-graph construction problem.
\end{lemma}
\begin{proof}
 This is based on Theorem~\ref{thm:equivalent-delivery-time+graph}
  and the approximation ratio of \textbf{HS} for the problem
  $1|r_j|L_{\max}$, where the proof can be found
  in~\cite{DBLP:journals/mor/HallS92}.
\end{proof}

The algorithm that has the best approximation ratio for the
delivery-time model of the problem $1|r_j|L_{\max}$ is a
polynomial-time approximation scheme (PTAS) developed by Hall and
Shmoys~\cite{DBLP:journals/mor/HallS92}.

\begin{lemma}
  \label{lemma:uniprocessor-PTAS}
  The dependency-graph construction problem admits a polynomial-time
  approximation scheme (PTAS), i.e., the approximation bound is
  $1+\epsilon$ under the assumption that $\frac{1}{\epsilon}$ is a
  constant for any $\epsilon >0$.
\end{lemma}

\section{Algorithms to Schedule Dependency Graphs}
\label{sec:heuristic-upper-bounds}

This section presents our heuristic algorithms to schedule the
dependency graph $G$ derived from
Algorithm~\ref{alg:graph-construction}. We first consider the special
case when there is a sufficient number of processors, i.e., $M \geq
N$.

\begin{lemma}
  \label{lemma:list-no-approximation}
  For a task set $\textbf{T}$, to be scheduled on $M$ identical processors, the
  makespan of the schedule which executes task $\tau_i$ on only one
  processor $i$ as early as possible by respecting to the precedence
  constraints defined in a given task
  dependency graph $G$
  is $len(G)$ if $M \geq N$. By definition, the schedule is a partitioned schedule
  for the given jobs and non-preemptive with respect to the subjobs.
\end{lemma}
\begin{proof}
 Since $M\geq N$, all the tasks can start their first non-critical
  sections at time $0$. Therefore, the critical section of task
  $\tau_i$ arrives exactly at time $C_{i,1}$. Then, the finishing time
  of the critical section of task $\tau_i$ is exactly the longest path
  in $G$ that finishes at the vertex representing
  $A_{i,1}$. Therefore, the makespan of such a schedule is
  exactly $len(G)$.
\end{proof}

For the remaining part of this section, we will focus on 
the other case
when $M < N$.  We will heavily utilize the concept of list schedules developed
by Graham~\cite{DBLP:journals/siamam/Graham69} and the extensions to
schedule the dependency graph $G$ derived from
Section~\ref{sec:heuristic-derive-G}. A list schedule works as
follows: Whenever a processor idles and there are subjobs eligible to
be executed (i.e., all of their predecessors in $G$ have finished),
one of the eligible subjobs is executed on the processor. When the
number of eligible subjobs is larger than the number of idle processors,
many heuristic strategies exist to decide which subjobs should be
executed with higher
priorities. Graham~\cite{DBLP:journals/siamam/Graham69} showed that the
list schedules can be generated in polynomial time and 
have a $2-\frac{1}{M}$ approximation ratio for the
scheduling problem $P|prec|C_{\max}$.

For the rest of this section, we will explain how to use or extend
list schedules to generate partitioned or semi-partitioned and
preemptive or non-preemptive schedules based on $G$.

\subsection{Semi-Partitioned Scheduling}

In a list schedule, since the subjobs of a task are scheduled
individually, a task in the generated list schedule may migrate
among different processors, thus representing a semi-partitioned
schedule. However, a subjob by default is non-preemptive in list schedules.

The following lemma is widely used in the literature for the list
schedules developed by
Graham~\cite{DBLP:journals/siamam/Graham69}. All the existing results
of federated scheduling, e.g.,
\cite{Li:ECRTS14,DBLP:conf/emsoft/Baruah15,Chen:RTS2016-Federated},
for scheduling sporadic dependent tasks (that are not due to
synchronizations) all implicitly or explicitly use the property in
this lemma. 
\begin{lemma}
  \label{lemma:list-upper}
  The makespan of a list schedule of a given task dependency graph $G$
  for task set $\textbf{T}$ on $M$ processors is at most
  $\frac{\sum_{\tau_i \in \textbf{T}} (C_{i,1}+A_{i,1}+C_{i,2}) - len(G)}{M} +  
  len(G)$.
\end{lemma}
\begin{proof}
  The original proof can be traced back to Theorem 1 by
  Graham~\cite{DBLP:journals/siamam/Graham69} in 1969. We omit the
  proof here as this is a standard procedure in the proof of list
  schedules for the scheduling problem $P|prec|C_{\max}$.
\end{proof}

\begin{lemma}
  \label{lemma:list-upper-approximation}
  If $len(G) \leq \alpha \times len(G^*)$ for a certain $\alpha \geq
  1$, the makespan of a list schedule of the task dependency graph $G$
  for task set $\textbf{T}$ on $M$ processors has an approximation
  bound of $1+\alpha-\frac{\alpha}{M}$ if $M < N$.
\end{lemma}
\begin{proof}
  Since $M < N$, the makespan of a list schedule of $G$, denoted
  as $L(List(G))$, is
  \begin{align}
    & L(List(G))\nonumber\\
    \overset{\mbox{\footnotesize
        Lemma~\ref{lemma:list-upper}}}{\leq}\;\;\;&\frac{(\sum_{\tau_i
        \in \textbf{T}}
      C_{i,1}+C_{i,2}+A_{i,1}) - len(G)}{M} +  len(G)    \nonumber\\
    = \qquad&\frac{\sum_{\tau_i \in \textbf{T}} C_{i,1}+C_{i,2}+A_{i,1}}{M} +  len(G)(1-\frac{1}{M})    \nonumber\\
    \overset{\footnotesize \mbox{assumption}}{\leq}\;\; &\frac{\sum_{\tau_i \in \textbf{T}}
      C_{i,1}+C_{i,2}+A_{i,1}}{M} + \alpha \times
    len(G^*)(1-\frac{1}{M} )  \nonumber \\
    \overset{\mbox{\footnotesize
        Theorem~\ref{theorem:lower-bound}}}{\leq}\;\; &
    (1+\alpha-\frac{\alpha}{M}) OPT
  \end{align}
\end{proof}

We now conclude the approximation ratio.
\begin{theorem}
  \label{theorem:list-upper-approximation}
  When applying \textbf{JKS} ($\alpha=2$, from
  Lemma~\ref{lemma:uniprocessor-EDD}), \textbf{Potts} ($\alpha=1.5$,
  from Lemma~\ref{lemma:uniprocessor-Potts}), \textbf{HS}
  ($\alpha=4/3$, from Lemma~\ref{lemma:uniprocessor-HS}), and PTAS
  ($\alpha=\epsilon$ for any $\epsilon > 0$, from
  Lemma~\ref{lemma:uniprocessor-PTAS}) to generate the task dependency
  graph $G$, the \tsocs{} Makespan problem admits polynomial-time
  algorithms to generate a semi-partitioned schedule
  that has an approximation ratio of
  \begin{equation}
    \begin{cases}
      \alpha & \mbox{ if } M \geq N\\
      1+\alpha-\frac{\alpha}{M}  & \mbox{ if } M < N
    \end{cases}
  \end{equation}
\end{theorem}
\begin{proof}
  The case when $M < N$ comes from
  Lemma~\ref{lemma:list-upper-approximation}. The case when $M \geq N$
  comes from Lemma~\ref{lemma:list-no-approximation} and the fact that 
  a partitioned schedule is also a semi-partitioned schedule by definition.
\end{proof}

The default list schedulers are non-preemptive in the subjob
level. However, it may be more efficient if the second non-critical
section of a task can be preempted by a critical section.
Otherwise, the processors may be busy executing second non-critical sections and a
critical section has to wait. As a result, not only this critical section itself
but also its successors in $G$ may be unnecessary postponed and therefore
increase the makespan. This problem can be handled  by preempting second
non-critical sections.
Allowing such preemption in the scheduler design can be achieved easily as follows:
\begin{itemize}
\item In the algorithm, the scheduling decision is made at a time $t$
  when there is a subjob eligible or finished.
\item Whenever a subjob representing a critical section is eligible,
  it can be assigned to a processor that executes a second
  non-critical section of a job by preempting that subjob. 
\end{itemize}

The makespan of the resulting schedule remains at most
$\frac{\sum_{\tau_i \in \textbf{T}} (C_{i,1}+A_{i,1}+C_{i,2}) -
  len(G)}{M} + len(G) $ as in Lemma~\ref{lemma:list-upper}. Therefore,
the approximation ratios in
Theorem~\ref{theorem:list-upper-approximation} still hold even if
preemption of the second non-critical sections is possible.

\subsection{Partitioned Scheduling}
\label{sec:partitioned}

In a partitioned schedule of the frame-based task set ${\bf T}$, all
subjobs of a task must be executed on the same processor. Therefore,
the list scheduling algorithm variant must ensure that once the first
subjob $C_{i,1}$ of task $\tau_i$ is executed on a processor, all
subsequent subjobs of task $\tau_i$ are tied to the
same processor in any generated list schedule.  Specifically, the
problem is termed as $P|prec, tied|C_{\max}$ in
Section~\ref{sec:scheduling-theory}.

A special case of $P|prec, tied|C_{\max}$ has been recently studied to
analyze OpenMP systems by Sun et al.~\cite{DBLP:conf/rtss/SunGWHY17}
in 2017.  They assumed that the synchronization subjob of a task
always takes place \emph{at the end of} the task.  Our dependency
graph $G$ unfortunately does not satisfy the assumption because the
synchronization subjob is in fact in the middle of a task. However, fixing this
issue is 
not difficult. We illustrate the key strategy by
using Fig.~\ref{fig:tied-example-dependency-graph}. The subgraph
$\bar{G}$ of $G$ that consists of only the vertices of the first
non-critical sections and the critical sections in fact satisfies
the assumption made by Sun et al.~\cite{DBLP:conf/rtss/SunGWHY17}.
Therefore, we can generate a multiprocessor schedule for the
dependency graph $\bar{G}$ on $M$ processors by using the BFS$^*$ algorithm (an extension of the breadth-first-scheduling algorithm) by Sun et al.~\cite{DBLP:conf/rtss/SunGWHY17}.  It can be
imagined that the subjobs that represent the second non-critical
sections $C_{i,2}$ are \emph{background} workload and can be executed
only at the end of the schedule or when the available idle time is
sufficient to complete $C_{i,2}$.

Alternatively, in order to improve the parallelism, another heuristic
algorithm can be applied where all the first non-critical sections
are scheduled before any of the critical sections using list
scheduling. Once the first non-critical section $C_{i,1}$ of task
$\tau_i$ is assigned on a processor, the remaining execution of task
$\tau_i$ is forced to be executed on that processor.

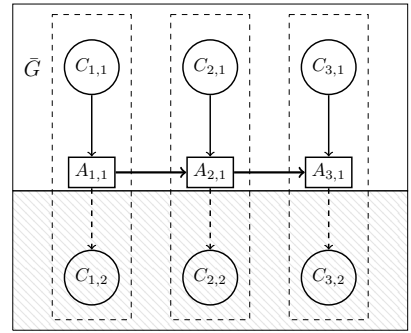
\begin{figure}[htb]
    \centering
  \scalebox{0.7}{
  \begin{tikzpicture}[x=0.75cm,auto, thick]
    
    \draw[solid, thin, pattern=north west lines, pattern color=black!10] (0,-1) -- (10,-1) -- (10,1.65) -- (0,1.65) -- (0,-1);
    \draw[solid, thin] (0, 1.65) -- (10,1.65) -- (10, 5.2) -- (0,5.2) -- (0,1.65);

    \node [](G) at (0.5, 4) {\large $\bar{G}$};

    \node [draw,circle](C11) at (2,4) {$C_{1,1}$};
    \node [draw,rectangle](A11) at (2,2) {$A_{1,1}$};
    \node [draw,circle](C12) at (2,0){$C_{1,2}$};
    
    \node [draw,circle](C21) at (5,4) {$C_{2,1}$};
    \node [draw,rectangle](A21) at (5,2) {$A_{2,1}$};
    \node [draw,circle](C22) at (5,0) {$C_{2,2}$};
    
    \node [draw,circle](C31)at(8,4){$C_{3,1}$};
    \node [draw,rectangle](A31)at(8,2){$A_{3,1}$};
    \node [draw,circle](C32)at(8,0){$C_{3,2}$};

    \draw[->] (C11) -- (A11);
    \draw[->] (C21) -- (A21);
    \draw[->] (C31) -- (A31);
    
    \draw[->,dashed] (A11) -- (C12);
    \draw[->,dashed] (A21) -- (C22);
    \draw[->,dashed] (A31) -- (C32);
  
    \draw[->,very thick] (A11) -- (A21);
    \draw[->,very thick] (A21) -- (A31);

    \draw[dashed, thin] (1,5) -- (3,5) -- (3,-0.8) -- (1,-0.8) -- (1,5);
    \draw[dashed, thin] (4,5) -- (6,5) -- (6,-0.8) -- (4,-0.8) -- (4,5);
    \draw[dashed, thin] (7,5) -- (9,5) -- (9,-0.8) -- (7,-0.8) -- (7,5);

  \end{tikzpicture}}     
    \caption{A schematic of a tied task dependency graph for a task set with
      one binary semaphore.}
    \label{fig:tied-example-dependency-graph}
  \end{figure}

If the second non-critical sections can be preempted, it can be
imagined that the subjobs that represent the second non-critical
sections $C_{i,2}$ are background workload and can be executed
whenever its processor idles and preempted by the first non-critical
sections or the critical sections on the processor. For completeness, 
we illustrate the algorithm in Algorithm~\ref{alg:tied-list-scheduling} in
the Appendix. 


\section{Timing Anomaly }
\label{sec:robustness-issues}

So far, we assume that $C_{i,1}$, $A_{i,1}$, and $C_{i,2}$ are exact
for a task $\tau_i$. However, the execution of a subjob of task
$\tau_i$ can be finished earlier than the worst case. It should be
noted that list schedules are in this case not sustainable, i.e., the reduction
of the execution time of a subjob can lead to a worse makespan due to the
well-known multiprocessor timing anomaly observed by
Graham~\cite{DBLP:journals/siamam/Graham69}. There are three ways to
handle such timing anomaly: 1) ignore the early completion and stick
to the offline schedule, 2) reclaim the unused time (slack)
carefully without creating timing anomaly, e.g., \cite{DBLP:conf/rtss/ZhuMC01}, or 3) use a
safe upper bound, e.g., Lemma~\ref{lemma:list-upper} to account for
all possible list schedules. Each of them has advantages and
disadvantages. It is up to the designers to choose whether they want
to be less effective (Option 1), pay more runtime overhead (Option
2), or be more pessimistic by taking always a safe upper bound (Option
3).

Due to multiprocessor timing anomaly, a dependency graph with a longer
critical path may have a better makespan in the resulting list
schedule.  Our approach can be easily improved by returning and
scheduling the intermediate dependency graphs in Algorithms Potts and
HS.

\section{Periodic Tasks with Different Periods}
\label{sec:periodic-tasks}

Our approach can be extended to periodic tasks with different periods
under an assumption that a binary semaphore is only shared among the
tasks that have the same period.
For each of the $z$ semaphores, a DAG is
constructed using 
Algorithm~\ref{alg:graph-construction}. Afterwards, the $z$ resulting DAGs
can be scheduled using any approach for multiprocessor DAG scheduling,
e.g., global scheduling~\cite{Lakshmanan:2010:SPR:1935940.1936239}, Federated Scheduling~\cite{Li:ECRTS14} as well as enhanced
versions like Semi-Federated
Scheduling~\cite{RTSS-2017-semi-federated} and Reservation-Based
Federated Scheduling~\cite{DBLP:journals/corr/abs-1712-05040}.

\section{Evaluations}
\label{sec:evaluations}

This section presents the evaluations of the proposed approach. We
will first explain how our approach can be
implemented by using existing routines in \litmus{} and provide
the measured overhead in \litmus{}. Then, we will demonstrate
the performance of the proposed approach by applying numerical
evaluations for different configurations.

\subsection{Implementations and Overheads}
\label{sec:implementation-litmus}

The hardware platform used in our experiments is a cache-coherent SMP,
consisting of two 64-bit Intel Xeon Processor E5-2650Lv4 running at 1.7
GHz, with 35 MB cache and 64 GB of main memory.  We have implemented
our dependency graph approach in
\litmus{},
in order to investigate the overheads. Both partitioned and
semi-partitioned scheduling algorithms presented in
Section~\ref{sec:heuristic-upper-bounds} have been implemented
in~\litmus{} under the plug-in Partitioned Fixed Priority (P-FP),
detailed in the Appendix. The patches of our implementation have been
released in~\cite{DGALITMUS}.

In Table~\ref{tab:overheads}, we report the following overheads of
different protocols, including the existing protocols
DPCP, and MPCP in \litmus{} and our implementation of 
the partitioned dependency graph
approach (PDGA) and the semi-partitioned
dependency graph approach (SDGA):
\begin{compactitem}
\item\textbf{CXS:} context-switch overhead.
\item\textbf{RELEASE:} time spent to enqueue a newly released job in a ready queue.
\item\textbf{SCHED2:} time spent to perform post context switch and management activities.
\item\textbf{SCHED:} time spent to make a scheduling decision (scheduler to find the next job).
\item\textbf{SEND-RESCHED:} inter-processor interrupt latency, including migrations.
\end{compactitem}
Table~\ref{tab:overheads} shows that the 
overheads of our approach and of other protocols implemented in \litmus{} are
comparable.

\begin{table}[]
	\scalebox{0.71}{\begin{tabular}{|c|c|c|c|c|}
		\hline
		Max.(Avg.) in $\mu s$& DPCP          & MPCP          & PDGA          & SDGA          \\ \hline
		CXS                 & 30.93 (1.51)  & 31.1 (0.67)   & 31.21 (0.71)  & 30.95 (1.54)  \\ \hline
		RELEASE             & 32.63 (3.96)  & 19.48 (3.91)  & 19.77 (4.03)  & 21.64 (4.3)   \\ \hline
		SCHED2              & 28.7 (0.18)   & 29.78 (0.15)  & 29.91 (0.16)  & 29.74 (0.2)   \\ \hline
		SCHED               & 31.43 (1.2)   & 31.38 (0.78)  & 31.4 (0.83)   & 31.26 (1.11)  \\ \hline
		SEND-RESCHED        & 47.01 (14.42) & 31.83 (3.45)  & 45.23 (4.33)  & 41.53 (7.24)  \\ \hline
	\end{tabular}}
	\caption{Overheads of different protocols in \litmus{}.}
	\label{tab:overheads}
	\vspace{-0.6cm}
\end{table}

\subsection{Numerical Performance Evaluations}
\label{sec:performance-evaluations}

We conducted evaluations with $M$ = 4, 8 and 16 processors. Depending on
$M$, we generate $1000$ task sets, each with $10M$ tasks.  For
each task set $\textbf{T}$, we generated synthetic tasks
with $\sum_{\tau_i \in \textbf{T}} C_{i,1}+C_{i,2}+A_{i,1}=M$ by
applying the RandomFixedSum method~\cite{emberson2010techniques}
and enforced that $C_{i,1}+C_{i,2}+A_{i,1} \leq 0.5$ for each task $\tau_i$.
 The
number of shared resources (binary semaphores) was set to $z \in \{4,
8, 16\}$. The length of the critical section $A_{i,1}$ 
is a fraction of the total execution
time $C_{i,1}+C_{i,2}+A_{i,1}$ of task $\tau_i$, depended on
$\beta \in \{5\%-50\%\}$.
The remaining part $C_i$ was split into $C_{i,1}$ and $C_{i,2}$ by  
drawing   $C_{i,1}$ randomly uniform from $[0, C_i]$ and setting $C_{i,2}$ 
 to $C_i-C_{i,1}$.

For a generated task set $\textbf{T}$,
we calculated a lower bound $LB$ on the optimal makespan  
based on Eq.~\eqref{eq:lb}.  Since deriving $len(G^*)$ is
computationally expensive, we used $\min_{\tau_i \in \textbf{T}}
C_{i,1} + \min_{\tau_i \in \textbf{T}} C_{i,2} + \max_{k=1,\ldots,z}
CriticalSum_k$ as a safe approximation for $len(G^*)$, where
$CriticalSum_k$ is the sum of the lengths of the critical sections
that share semaphore $s_k$.  If the relative deadline of the task set
is less than $LB$, the task set is not schedulable by any algorithm.
We compare the performance of different algorithms according to the
\emph{acceptance ratio} by setting the relative deadline $D=T$ in the
range of $[LB, 1.8LB]$.  We name the developed algorithms using the
following rules: 1) \emph{JKS/POTTS} in the first part: using the
extended Jackson's rule or Potts to construct the dependency
graph;\footnote{We did not implement
  Lemma~\ref{lemma:uniprocessor-PTAS} due to the complexity issue.
  Algorithm HS in general has similar performance to POTTS. } 2)
\emph{SP/P} in the second part: semi-partitioned or partitioned
scheduling algorithm is applied\footnote{In
  Section~\ref{sec:partitioned}, we presented two strategies for task
  partitioning: one is based on \cite{DBLP:conf/rtss/SunGWHY17}
  (detailed in Appendix) and another is a simple heuristic by
  performing the list scheduling algorithm based on the first
  non-critical sections. In all the experiments regarding partitioned
  scheduling, we observed that the latter (i.e., the simple heuristic)
  performed better. All the presented results for partitioned scheduling are
  therefore based on the simple heuristic.  }; 3) \emph{P/NP} in the
third part: preemptive or non-preemptive for the second non-critical
sections.

\begin{figure}[t]
	\includegraphics[width=1\linewidth]{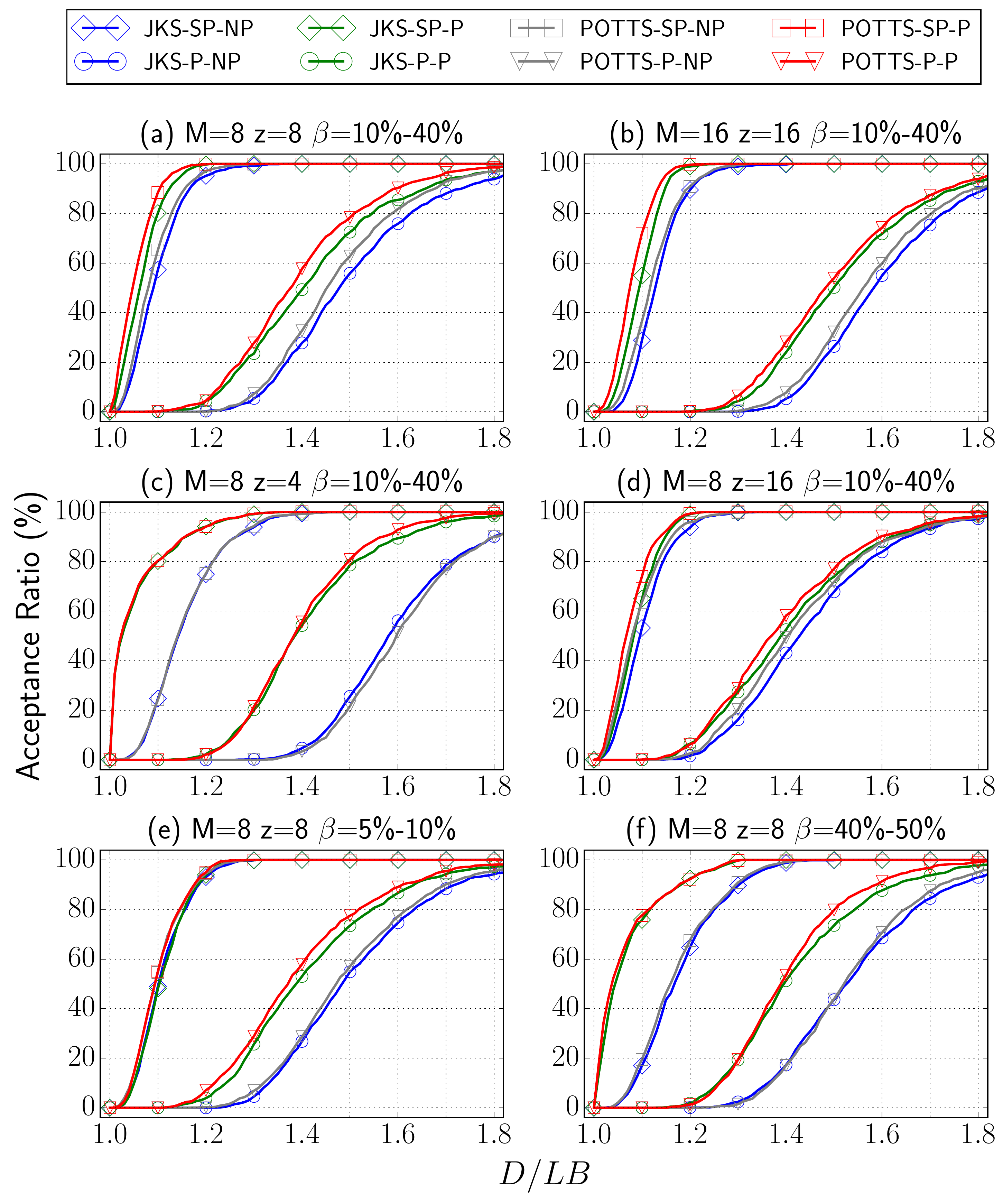}
	\vspace{-0.8cm}
	\caption{Comparison of different approaches with different deadlines.}
	\vspace{-0.3cm}
	\label{fig:ratio}
\end{figure}


We evaluated all 8 combinations under different settings as shown in
Fig.~\ref{fig:ratio}. Due to space limitation, only a subset
of the results is presented. In general, 
the semi-partitioned scheduling algorithms clearly outperform the
partitioned strategies, independently from the algorithm used to construct the
dependency graph. In addition, the preemptive scheduling policy with respect to the second computation segment is
superior to the non-preemptive strategy and POTTS (usually) performs slightly
better than JKS.
We analyze the effect of the three parameters individually by changing: 
\begin{enumerate}
  \item \textbf{$\boldsymbol{M = z \in \{8, 16\}}$}
  (Fig.~\ref{fig:ratio}(a)~and~Fig.~\ref{fig:ratio}(b)): increasing $z$ and $M$  also slightly
increases the difference between the semi-partitioned and the partitioned
approaches.
\item \textbf{$\boldsymbol{z}$ for a fixed $\boldsymbol{M}$}, i.e., $z \in \{4, 8, 16\}$ and $M = 8$ 
(Fig.~\ref{fig:ratio}(c), Fig.~\ref{fig:ratio}(a),~and
Fig.~\ref{fig:ratio}~(d)): when the number of resources is decreased compared to
the number of processors, the performance gap between preemptive and
non-preemptive scheduling increases.
\item \textbf{Workload of Shared Resources,
i.e., \\$\boldsymbol{\beta \in \{
[5\%-10\%], 
[10\%-40\%], 
[40\%-50\%]\}}$}
\\(Fig.~\ref{fig:ratio}(e),
Fig.~\ref{fig:ratio}~(a),~and~Fig.~\ref{fig:ratio}~(f)): if the workload of the
critical sections is increased, the difference between preemptive and
non-preemptive scheduling approaches is more significant.
\end{enumerate}

We also compare our approach with the Resource Oriented Partitioned (ROP)
scheduling with release enforcement by von der Br\"uggen et
al.~\cite{RTNS17-resource} which is designed to schedule periodic tasks with one
critical section on a multiprocessor platform.  The concept of the ROP is to have a
resource centric view instead of a processor centric view. The algorithm 1)
binds  the critical sections of the same resource to the same processor, thus 
enabling well known uniprocessor protocols like PCP to handle  the
synchronization,  and 2) schedule the non-critical sections on the remaining 
processors using a state-of-the-art scheduler for segmented self-suspension
tasks,  namely SEIFDA~\cite{Bruggen16RTNS}.
 Among the methods in~\cite{RTNS17-resource}, we evaluated FP-EIM-PCP 
 (under fixed-priority scheduling) and EDF-EIM-PCP (under dynamic-priority scheduling). 
 It has been shown in~\cite{RTNS17-resource} that
EDF-EIM-PCP dominates all existing methods. 
We performed another set of evaluations by adopting
aforementioned settings and testing the utilization level 
in a step of
$5\%$, where the utilization of a task set $\textbf{T}$ is
$\sum_{\tau_i\in \textbf{T}}\frac{C_{i,1}+C_{i,2}+A_{i,1}}{T_i}$.
 Fig.~\ref{fig:sched} presents the evaluation results. Due to
space limitation, only a subset of the results is presented, but the
others have very similar curve tendencies.  For readability, we only
select two combinations in our proposed approach that outperform the
others. The results in Fig.~\ref{fig:sched} show that for
frame-based tasks, our approach outperforms
ROP significantly.  We note that Fig.~\ref{fig:sched} is only for
frame-based tasks, and the results for periodic task
systems discussed in Section~\ref{sec:periodic-tasks} are further
presented in Appendix. 

\begin{figure}[t]
	\centering
	\includegraphics[height=1.5in]{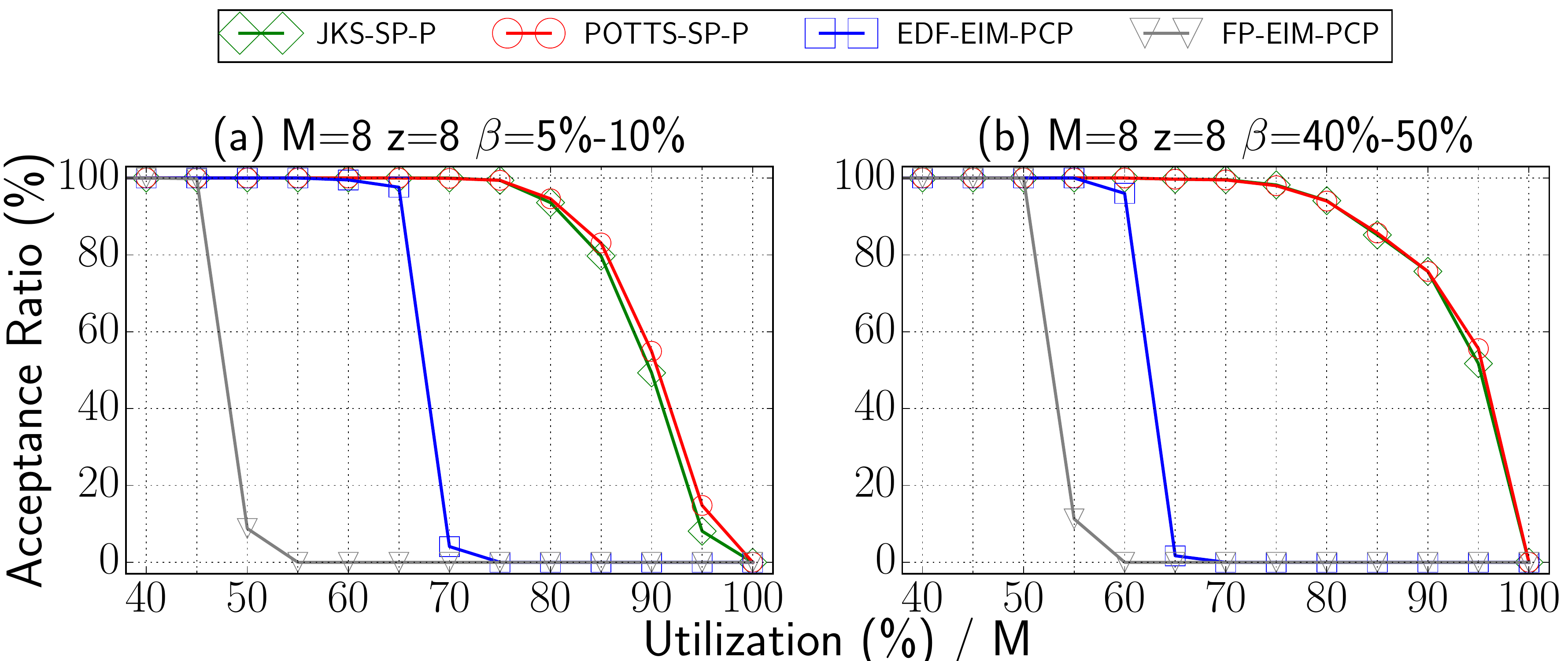}
	\vspace{-0.5cm}
	\caption{Schedulability of different approaches for
          frame-based task sets.}
	\vspace{-0.3cm}
	\label{fig:sched}
\end{figure}

\vspace{-0.15in}
\section{Conclusion} 
\label{sec:conclusion}
\vspace{-0.05in}

This paper tries to answer a few fundamental questions when real-time
tasks share resources in multiprocessor systems. Here is a short
summary of our findings:
\begin{itemize}
\item The fundamental difficulty is mainly due to the sequencing of
  the mutual exclusive accesses to the share resources (binary
  semaphores). Adding more processors, removing periodicity and job
  recurrence, introducing task migration, or allowing preemption does
  not make the problem easier from the computational
  complexity perspective.

\item The performance gap of partitioned and semi-partitioned
  scheduling in our study is mainly due to the capability to schedule
  the subjobs constrained by the dependency graph. Although
  partitioned scheduling may seem much worse than semi-partitioned
  scheduling in our evaluations, this is mainly due to the lack of
  understanding of the problem $P|prec, tied|C_{\max}$ in the
  literature. Further explorations are needed to understand these
  scheduling paradigms for a given dependency graph.

\item The dependency graph approach is not work-conserving for the
  critical sections, since a critical section may be ready but not
  executed due to the artificially introduced precedence
  constraints. Existing multiprocessor synchronization protocols
  mainly assume work-conserving for granting the accesses of the
  critical sections via priority boosting. Our study reveals a potential to consider
  cautious and non-work-conserving synchronization protocols in the
  future.
\end{itemize}

\clearpage 
\noindent{\textbf{Acknowledgement}}: {This paper is supported by DFG, as  part of the
Collaborative Research Center SFB876, project A3 and B2
{(http://sfb876.tu-dortmund.de/)}.} The authors thank Zewei Chen and Maolin Yang
for their tool SET-MRTS (Schedulability Experimental Tools for
Multiprocessors Real Time Systems, 
https://github.com/RTLAB-UESTC/SET-MRTS-public) to evaluate the LP-GFP-FMLP,
LP-PFP-DPCP, LP-PFP-MPCP, GS-MSRP, and LP-GFP-PIP in Fig.~\ref{fig:periodic-ratio}.

\section*{Appendix}
\begin{proofAppendix}{Theorem~\ref{theorem:non-optimal-approach}}
    Due to the design of the task set, there are only $N$ different
    dependency graphs, depending on the 
    position of $\tau_1$ in the execution order.  
    Suppose that the critical section of task $\tau_1$ is
    the $j$-th critical section in the dependency graph. It can be
    proved that the critical path of this dependency graph is $j\delta
    + Q + N\delta$.  
    We \emph{sketch} the proof:
    \begin{itemize}
    \item The non-critical section $C_{1,2}$ must be part of the
      critical path since $C_{1,2} = \frac{Q}{M} +N\delta$, which is
      greater than any $(N-1)A_{i,1} + C_{i,2}$ for any $i
      =2,3,\ldots, N-1$.
    \item The longest path that ends at the vertex representing
      $A_{1,1}$ has 1) one non-critical section, 2) $j-1$ critical
      sections from $\tau_i$ for $i=2,3,\ldots,N$, and 3) 1 critical
      section from task $\tau_1$. Therefore, this length is $\delta +
      (j-1)\delta +   Q - \frac{Q}{M} = j\delta + Q - \frac{Q}{M}$.
    \item Combining the two scenarios, we reach the conclusion.
    \end{itemize}

    Therefore, the dependency graph $G^*$ that has the minimum
    critical path length is the one where $\tau_1$'s critical section is
    the first one     
    among the $N$ critical sections. The optimal schedule of
    the dependency graph $G^*$ on $M$ processors has the following
    properties:
    \begin{itemize}
    \item Task $\tau_1$ finishes its critical section at time $\delta+Q-\frac{Q}{M}$.
    \item Before time $\delta+Q-\frac{Q}{M}$, none of the second
      non-critical sections is executed. Therefore, the makespan of
      any feasible schedule $S(G^*)$ of $G^*$ on $M$ processors is
      \begin{align*}
  L(S(G^*)) \geq      &~\delta+Q-\frac{Q}{M}+\sum_{i=1}^{N}\frac{C_{i,2}}{M} \\
 =&~
        \delta+Q -\frac{Q}{M} + \frac{(M^2-M+1) \frac{Q}{M} +
          N\delta}{M}\\
=&~ \left(1+\frac{N}{M}\right)\delta+\left(2- \frac{2}{M} +
\frac{1}{M^2}\right)Q
      \end{align*}
    \item Moreover, when the scheduling policy is either
      semi-partitioned or partitioned scheduling, by the pigeon hole
      principle, at least one processor must execute
      $\ceiling{\frac{N}{M}}$ of the $N$ second non-critical
      sections no earlier than $\delta+Q-\frac{Q}{M}$. Therefore, the makespan of a feasible semi-partitioned
      or partitioned schedule $S_p$ of $G^*$ on $M$ processors is
      \begin{align*}
L(S_p(G^*)) \geq        &~\delta+Q-\frac{Q}{M}+\ceiling{\frac{N}{M}}
\frac{Q}{M}\\
 =&     ~\delta+Q -\frac{Q}{M} +\ceiling{M-1+\frac{1}{M}} \frac{Q}{M}  \\
  =&    ~\delta+Q -\frac{Q}{M} +M \frac{Q}{M}  \\
= &~ \delta+\left(2- \frac{1}{M} \right)Q
      \end{align*}
    \end{itemize}

\noindent    We can have another feasible partitioned schedule $S^*$:
    \begin{itemize}
    \item The first non-critical section $\tau_1$ is executed on
      processor $M$, and the first non-critical sections of the other
      $N-1$ tasks are executed on the first $M-1$ processors based on
      list scheduling.  All the first non-critical sections finish no
      later than $M\delta$. Each of the first $M-1$ processors
      executes \emph{exactly} $M$ tasks since there are $N-1=M(M-1)$
      tasks with identical properties on these $M-1$
      processors.
   \item The critical sections of tasks $\tau_N,\tau_{N-1},
     \ldots, \tau_1$ are executed sequentially by following the above
     reversed-index order on the same processor of the corresponding first
     non-critical sections, starting from time $M\delta$.
    \item At time $M\delta+N\delta$, all the second non-critical
      sections of $\tau_2, \ldots, \tau_N$ are eligible to be
      executed. We execute them in parallel on the first $M-1$
      processors by respecting the partitioned scheduling
      strategy. That is, each of the first $M-1$ processors
      executes \emph{exactly} $M$ tasks with $C_{i,2}= Q/M$.
The makespan of these $N-1$ tasks is $(N+M)\delta +
      \frac{(N-1)\frac{Q}{M}}{M-1} = (N+M)\delta + Q$.
    \item At time $M\delta+N\delta$, the critical section of $\tau_1$
      starts its execution on processor $M$. Furthermore, at time
      \mbox{$(N+M)\delta+Q-\frac{Q}{M}$}, the second non-critical section of
      $\tau_1$ is executed on processor $M$ and it is finished at time
      \mbox{$(N+M)\delta+Q +N\delta = (2N+M)\delta + Q$}.
    \item 
    As a result, the makespan of the above partitioned schedule
      $S^*$ is \emph{exactly} $(2N+M)\delta + Q$.
    \end{itemize}

    Therefore, the approximation bound of the optimal task dependency
    graph approach is at least $\frac{L(S(G^*))}{L(S^*)}$ under any
    scheduling paradigm and is at least $\frac{L(S_p(G^*))}{L(S^*)}$
    under partitioned or semi-partitioned scheduling paradigm. We
    reach the conclusion by taking $\delta\rightarrow 0$.  
\end{proofAppendix}


\textbf{Pseudo-code of the Partitioned Preemptive Scheduling in
  Section~\ref{sec:partitioned}} 
For notational brevity, we define
two vertices $v_{i,1}$ and $v_{i,3}$ to represent the first and second
non-critical sections of task $\tau_i$ and $v_{i,2}$ to represent the
critical section of task $\tau_i$. Let $\textbf{T}_m$ be the set of
tasks in $\textbf{T}$ assigned to processor $m$ for
$m=1,2,\ldots,M$. The pseudo-code is listed in
Algorithm~\ref{alg:tied-list-scheduling}. It consists of three blocks:
initialization from Line 1 to Line 4, scheduling of the first
non-critical sections and the critical sections of the tasks according
to $\bar{G}$ from Line 5 to Line 23, and scheduling of the second
non-critical sections of the tasks from Line 24 to Line 28.

The first block is self-explained in
Algorithm~\ref{alg:tied-list-scheduling}. We will focus on the
second and third blocks of Algorithm~\ref{alg:tied-list-scheduling}.  
Our scheduling algorithm executes the first non-critical sections and
the critical sections non-preemptively. Whenever a
subjob finishes at time $t$, we examine the following scenarios on
each processor $m$ for $m=1,2,\ldots,M$:
\begin{itemize}
\item If there is a pending critical section on processor $m$ that is
  eligible at time $t$ according to the dependency graph $G$, we would
  like to execute the critical section as soon as possible. Therefore,
  this critical section is executed as soon as it is eligible and the
  processor idles (i.e., Lines 12-13).
\item Else if there is a task in $\textbf{T}_m$ in which its first
  non-critical section has not finished yet at time $t$, we would like
  to execute it (Lines 14-15).
\item Otherwise, there is no eligible subjob to be executed at time
  $t$. If there is still an unassigned task, we select one and assign
  it to processor $m$ by starting its first non-critical section at
  time $t$ (Lines 16-19).
\end{itemize}
In all the above steps, task $\tau_i$ can be arbitrarily selected if
there are multiple tasks satisfying the specified conditions. 
We note that the schedule is in fact \emph{offline}. Therefore, after
we finish the schedule of the first non-critical sections and the
critical sections, in the third block in
Algorithm~\ref{alg:tied-list-scheduling}, we can pad the idle time of
the schedule on a processor $m$ with the second non-critical sections
assigned on processor $m$, starting from time $0$. The only attention
is not to start earlier than the finishing time of its critical
section. Of course, to minimize the makespan, we should always pad the
idle time as early as possible.

\begin{algorithm}[t]
  \caption{Tied List-Scheduling (Partitioned Preemptive)}
  \label{alg:tied-list-scheduling} 
  \begin{algorithmic}[1]
  \footnotesize
  \INPUT{$G, \textbf{T}, M$ with $|\textbf{T}| > M$};

  \STATE{$current \leftarrow 0$};
  \STATE{assign \emph{one} task $\tau_i$ in $\textbf{T}$ to task set
    $\textbf{T}_m$ to be executed on processor $m$};

  \STATE{$\textbf{T} \leftarrow \textbf{T} \setminus \cup_{m=1}^{M} \textbf{T}_m$};
  \STATE{execute $v_{i,1}$ of the unique task $\tau_i$ in $\textbf{T}_m$
    on processor $m$ from time $0$,  i.e.,
    $\rho(t, m) \leftarrow \tau_i$ for $t \in [0, C_{i,1})$, for each
    $m=1,2,\ldots,M$};

  \WHILE {$\exists \tau_i$ such that $v_{i,2}$  has not finished yet at time $current$}
    \STATE {let $t$ be the minimum time instant greater than $current$
    such that the schedule finishes a subjob at time $t$};
  \STATE{$current \leftarrow t$};

    \FOR {$m=1,2,\ldots,M$}
      \IF{processor $m$ is busy executing a subjob at time $t$}
      \STATE{continue;}
      \ELSIF{processor $m$ idles (or just finishes a subjob) at time $t$}
      \IF {$\exists \tau_i \in \textbf{T}_m$, in which $v_{i,2}$ has not
      finished yet and $v_{i,2}$ is eligible
      according to $G$ at time $t$}
        \STATE{execute $\tau_i$'s critical section from time
          $t$ to $t+A_{i,1}$ non-preemptively on processor $m$, i.e., $\rho(\theta, m) \leftarrow \tau_i$ for $\theta \in [t, t+A_{i,1})$};
    \ELSIF{$\exists \tau_i \in \textbf{T}_m$, in which $v_{i,1}$ has not
      finished yet at $t$}
         \STATE{execute $v_{i,1}$ from time
          $t$ on $t+C_{i,1}$ processor $m$, i.e., $\rho(\theta, m) \leftarrow \tau_i$ for $\theta \in [t, t+C_{i,1})$};
     \ELSIF{$\textbf{T}$ is not empty}
          \STATE{select a task $\tau_i$ and remove $\tau_i$ from
            $\textbf{T}$, i.e., $\textbf{T} \leftarrow \textbf{T} \setminus
            \setof{\tau_i}$};
          \STATE{assign task $\tau_i$ to processor
            $m$, i.e., $\textbf{T}_m \leftarrow \textbf{T}_m \cup
            \setof{\tau_i}$};
          \STATE{execute $v_{i,1}$ from time $t$ to $t+C_{i,1}$ on
            processor $m$,  i.e., $\rho(\theta, m) \leftarrow \tau_i$ for $\theta \in [t, t+C_{i,1})$};
    \ENDIF
    \ENDIF
    \ENDFOR
  \ENDWHILE
   \FOR {$m=1,2,\ldots,M$}
      \FOR {each task $\tau_i$ in $\textbf{T}_m$}
      \STATE{schedule the second non-critical section $v_{i,3}$ of task $\tau_i$
        as background workload with the lowest priority preemptively
        as early as possible but no earlier than the finishing time of its critical section};
      \ENDFOR
   \ENDFOR
  \end{algorithmic}
  \end{algorithm}
  
\noindent\textbf{Implementation in \litmus}    
To force the tasks to follow the pre-defined order to execute the
critical sections, we added several elements into the $rt\_ params$
structure which is used to define the property for each
task, i.e., priority, period, execution time, etc. Two parameters are
added: 1) \emph{$rt\_order$} to define the order of the task to execute
the critical section, and 2) \emph{$rt\_total$} to define the number of
the tasks that shared the same resource.  To implement the
binary semaphores under the dependency graph approach, we created two new
structures, \emph{$pdga\_semaphore$} for the partitioned dependency graph
approach (PDGA), and \emph{$sdga\_semaphore$} for the semi-partitioned
dependency graph approach (SDGA). In these structures, one
parameter is defined to control the order of the execution named
\emph{$current\_serving\_ticket$}. When a task requests the resource,
it will compare its \emph{$rt\_order$} with the semaphore's
\emph{$current\_serving\_ticket$}, if they are equal, the task will be
granted to access the resource and start its critical section; if
not, the task will be added to the wait-queue, which is sorted by
the tasks' parameter \emph{$rt\_order$}. Once a task has finished its
critical section, it will increase the semaphore's current
serving ticket by $1$, and check the head of the wait-queue do the comparison
again. Once the \emph{$current\_serving\_ticket$} reaches to the
\emph{$rt\_total$}, which means one dependency graph has finished its
execution of the critical sections, then the parameter \emph{$current\_serving\_ticket$} will
be reset to $0$ to start the next iteration. The only difference between
PDGA and SDGA is that we added the migration function for SDGA to support
the semi-partitioned algorithm.

\noindent\textbf{Evaluations for Periodic Task Sets}   
We also performed evaluations for periodic task systems, when a binary
semaphore is only shared by the tasks with the same period described
in Section~\ref{sec:periodic-tasks}. We used
similar configurations as in Section~\ref{sec:performance-evaluations} to
generate the task sets. For the tasks that share the same semaphore,
they have the same period in the range of $[1, 10]$. The following
algorithms were evaluated:
\begin{compactitem}
	\item LP-GFP-FMLP~\cite{block-2007}: a linear-programming-based
	(LP) analysis for global FP scheduling using the FMLP~\cite{block-2007}. 
	\item LP-PFP-DPCP~\cite{bbb-2013}: LP-based analysis for
	partitioned FP and DPCP~\cite{DBLP:conf/rtss/RajkumarSL88}. Tasks are assigned
	using Worst-Fit-Decreasing (WFD) as proposed in~\cite{bbb-2013}.
	\item LP-PFP-MPCP~\cite{bbb-2013}: LP-based analysis for
	partitioned FP using MPCP~\cite{Rajkumar_1990}. Tasks are partitioned according to WFD
	as proposed in~\cite{bbb-2013}. 
	\item GS-MSRP~\cite{wieder-2013}: the Greedy Slacker (GS)
	partitioning heuristic with the spin-based locking protocol
	MSRP~\cite{DBLP:conf/rtss/GaiLN01} under Audsley's Optimal Priority Assignment~\cite{Audsley1991aOPA}.
	\item LP-GFP-PIP: LP-based global FP scheduling using the
	Priority Inheritance Protocol (PIP)~\cite{DBLP:conf/rtss/EaswaranA09}. 
	\item FP-EIM-PCP~\cite{RTNS17-resource}: The ROP scheduling under fixed-priority scheduling and release enforcement.
	\item EDF-EIM-PCP~\cite{RTNS17-resource}: The ROP scheduling under dynamic-priority scheduling and release enforcement.
	\item POTTS-SF: Our approach by applying algorithm Potts for generating $G$ and
          semi-federated scheduling
          in~\cite{RTSS-2017-semi-federated}.
	\item JKS-SF: Our approach by applying algorithm JKS for generating $G$ and
          semi-federated scheduling  in~\cite{RTSS-2017-semi-federated}. 
\end{compactitem}
For one evaluation point, $100$ synthetic task sets were generated and
tested. Only a subset of the results is presented in
Fig.~\ref{fig:periodic-ratio} due to space limitation, and
LP-PFP-MPCP is not presented for better readability since it performs
the worst for the evaluations in Fig.~\ref{fig:periodic-ratio}.  The figure
clearly shows that POTTS-SF and JKS-SF significantly outperform the
other approaches.

\begin{figure}[t]
	\centering
        \vspace{-0.25in}
	\includegraphics[width=\linewidth]{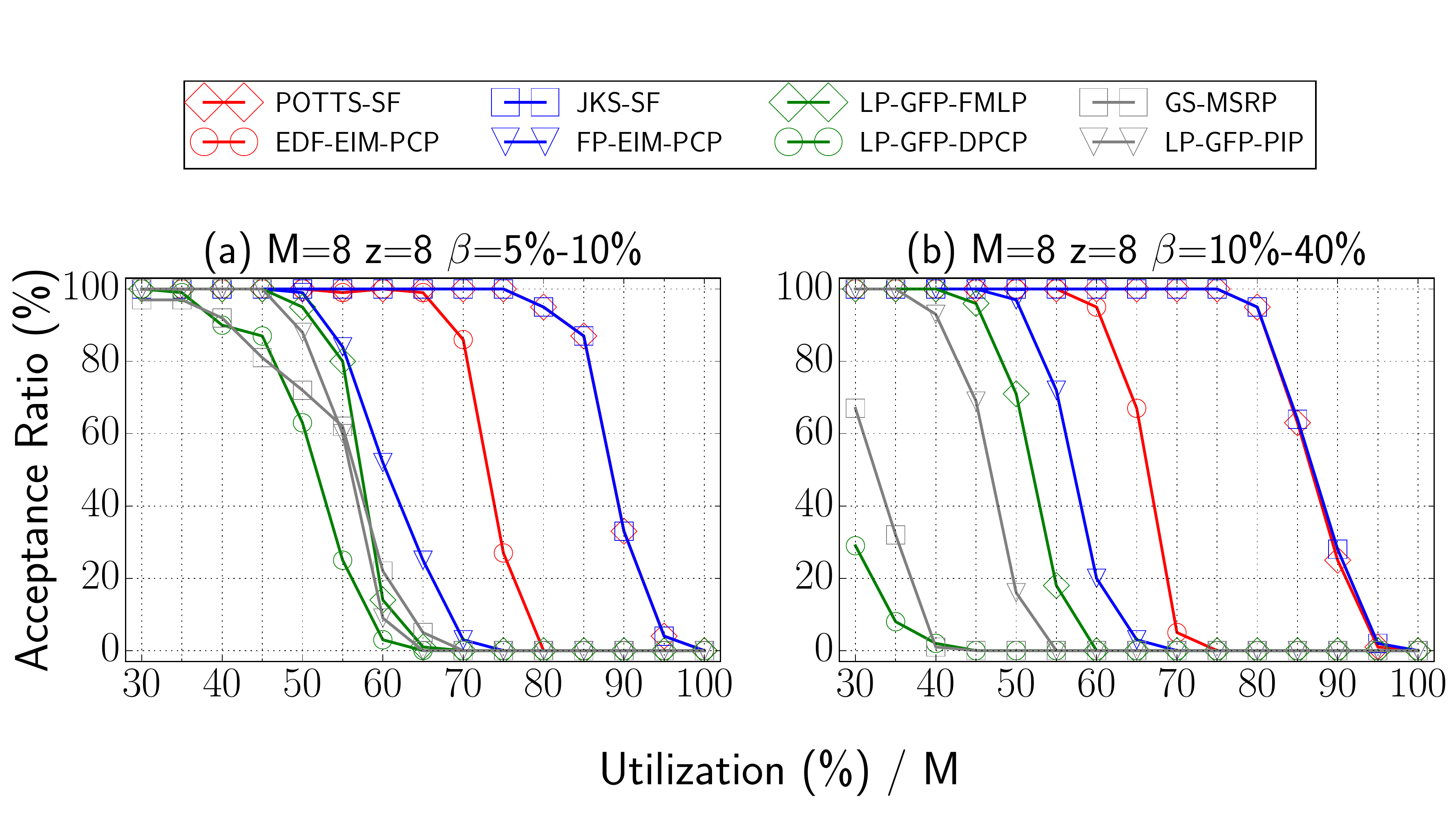}
        \vspace{-0.4in}
	\caption{Comparison of different approaches for periodic task sets. }
	\label{fig:periodic-ratio}
\end{figure}

\bibliographystyle{abbrv} \bibliography{real-time}

\end{document}